\documentclass[a4,11pt]{article}
\usepackage{graphicx}
\usepackage{amsmath}
\usepackage{amssymb}
\usepackage{bbm}
\usepackage{caption}
\usepackage{multirow}
\usepackage{geometry}
\usepackage{tikz}
\usepackage{verbatim}
\usetikzlibrary{shapes}
\usepackage{amsthm}
\usepackage[bookmarks=true,colorlinks=true]{hyperref}
\usepackage{xspace}
\newtheorem{theorem}{Theorem}[section]
\newtheorem{lemma}{Lemma}[section]
\newtheorem{proposition}{Proposition}[section]

\newtheorem{remark}{Remark}[section]
\newtheorem{assumption}{Assumption}[section]

\numberwithin{equation}{section}
\numberwithin{figure}{section}
\numberwithin{table}{section}
\usepackage{float}
\usepackage{indentfirst}
\usepackage{caption}
\usepackage{subfigure}
\allowdisplaybreaks
\usepackage{epstopdf} 

\begin{document}
	\title{\bf Sensitivity of Optimal Retirement Problem to Liquidity Constraints} 
	\date{}
	\author{\sffamily  Guodong Ding$^{1}$, Daniele Marazzina$^{1,2}$\\
		{\sffamily\small $^1$ Department of Mathematics, Politecnico di Milano}\\ {\sffamily \small Piazza Leonardo da Vinci 32, I-20133, Milano, Italy}\\
		{\sffamily\small $^2$ Corresponding Author, daniele.marazzina@polimi.it}}
	\maketitle
	
	{\noindent\small{\bf Abstract:}
	 In this work we analytically solve an optimal retirement problem, in which the agent optimally allocates the risky investment, consumption and leisure rate to maximise a gain function characterised by a power utility function of consumption and leisure, through the duality method. We impose different liquidity constraints over different time spans and conduct a sensitivity analysis to discover the effect of this kind of constraint.  }
	
	\vspace{1ex}
	{\noindent\small{\bf Keywords:} Liquidity Constraints, Retirement Stopping Time, Consumption-Portfolio-Leisure Controls, Duality Method, Variational Inequalities }

	\section{Introduction}
	  We study a stochastic control problem involving the consumption-portfolio-leisure policy and the optimal stopping time of retirement. By determining the continuous and stopping regions of the corresponding optimal stopping time problem, we prove that the optimal retirement time is the first hitting time of the wealth process $X(t)$ upward to a critical wealth boundary. We implement different liquidity constraints over different time spans, which are $X(t)\ge R_{pre}$ and $X(t)\ge R_{post}$ separately for pre- and post-retirement periods. The numerical analysis shows that the wealth boundary triggering the retirement is decreasing to $R_{pre}$ but increasing to $R_{post}$. The additional retirement option impels the agent to consume less and invest more as the wealth approaches the retirement boundary, and this incentive becomes weaker as $R_{pre}$ decreases.
	 
	  The considered retirement mechanism is directly referred to \cite{farhi2007saving,choi2008optimal}. More precisely, \cite{farhi2007saving} studied the optimal retirement model regarding the consumption-portfolio-leisure strategy, in which the leisure rate is limited to the binomial choice. \cite{choi2008optimal} investigated a more complex optimization problem that endows the agent the flexibility in labour supply in the context of retirement planning. We extend their research and adopt a different utility function, a power utility function, as in \cite{shim2014optimal}, instead of the Constant Elasticity of Substitution (CES) function. Additionally, compared to \cite{choi2008optimal}, other extensions are i) the introduction of a continuous debt repayment the agent should face, ii) the different liquidity constraints before and after retirement, which is the main contribution of this work.

	 \section{Problem Formulation}
	 
	  We deal with a financial market  in which two kinds of investment are provided: the money market, concerning a fixed risk-free rate $r>0$, and a risky asset, which dynamics is described by the stochastic differential equation $d S(t)=\mu S(t)d t+\sigma S(t)dB(t), \,S(0)=S_0$, with $\mu$ and $\sigma$ representing the constant drift and diffusion coefficients. $B(t)$ represents a standard Brownian motion on the filtered probability space $(\Omega,\mathcal{F},\mathbb{P})$, and $\{\mathcal{F}_{t},0\leq t<\infty\}$ is the augmented natural filtration on $B(t)$. Moreover, by introducing the market price of risk as $\theta\triangleq\frac{\mu-r}{\sigma}$, we can define the state-price density process as $H(t)\triangleq\xi(t)\tilde{Z}(t)$ following \cite{karatzas1998methods}, where $\xi(t)\triangleq e^{-rt}$ and $\tilde{Z}(t)\triangleq e^{-\frac{\theta^{2}}{2}t-\theta B(t)}$ indicate the discount process and an exponential martingale, respectively. Then we define the equivalent martingale measure $\tilde{\mathbb{P}}$ by $\tilde{\mathbb{P}}(A)\triangleq\mathbb{E}\left[\tilde{Z}(t)\mathbb{I}_{A}\right]$, $\forall A\in \mathcal{F}_{t}$. Based on the Girsanov Theorem, a standard Brownian motion under $\tilde{\mathbb{P}}$ measure can be defined as $\tilde{B}(t)\triangleq B(t)+\theta t$, $\forall t\ge 0$.
	 
	  We now describe the optimization problem. The agent needs to optimally allocate the consumption $c(t)$, the amount of money for the risky investment $\pi(t)$ and the leisure rate $l(t)$. The sum of labour and leisure rates equals the constant $\bar{L}$. Furthermore, denoting the retirement time as $\tau$, the retirement mechanism can be elaborated as: $0\leq l(t)\leq L<\bar{L}$ on $0\leq t\leq\tau$, i.e., the leisure rate, as the complement of labour rate, is upper bounded for keeping the employment state; and $l(t)\equiv\bar{L}$ on $t>\tau$, since the agent enjoys the entire leisure $\bar{L}$ after declaring retirement. Then the dynamics of the wealth process $X(t)$, i.e., the state variable of the optimization, is
	 $$d X(t)=\left[r X(t)+\pi(t)(\mu-r)-c(t)-d+w(\bar{L}-l(t))\right]d t+\sigma\pi(t)dB(t),\ \forall t\ge 0,$$
	 $d$ and $w$ are the constant debt repayment and the wage rate, respectively. The initial wealth is $X(0)=x$. \\
	 The considered optimal retirement problem $(P)$ is
	 \begin{equation*}
	 V(x)\triangleq\sup_{(\tau,\{c(t),\pi(t),l(t)\})\in\mathcal{A}(x)}J(x;c,\pi,l,\tau)=\sup_{(\tau,\{c(t),\pi(t),l(t)\})\in\mathcal{A}(x)} \mathbb{E} \left[\int_{0}^{\infty}e^{-\gamma t}u(c(t),l(t))d t\right], \tag{$P$}
	 \end{equation*}
	 in which $\gamma$ is the subjective discount rate, and the utility is characterized by a power function $$u(c,l)\!=\!\frac{\left(c^{\delta}l^{1-\delta}\right)^{1-k}}{\delta(1-k)}, \ 0\!<\!\delta\!<\!1,\ k\!>\!1.$$ The admissible control set $\mathcal{A}(x)$ follows the standard definition, e.g., \cite[Definition 2.1]{2021arXiv210615426D}, imposing liquidity constraints: $X(t)\ge R_{pre}$ for $0\leq t<\tau$, $X(\tau)\ge R_{pre}\vee R_{post}$, and $X(t)\ge R_{post}$ for  $t>\tau$ a.s.. Notice that we must impose $R_{pre}\ge\frac{d-w\bar{L}}{r}$ and $R_{post}\ge\frac{d}{r}$ to have the existence of an admissible solution, where $\frac{d-w\bar{L}}{r}$ represents the discounted value of the full debt repayment minus the maximum amount to borrow against the future labour income (in the pre-retirement period). 

	\section{Solution of Optimization Problem}
	Defining $J_{\scriptscriptstyle PR}(X(\tau);c,\pi)\triangleq\mathbb{E} \left[\left.\int_{\tau}^{\infty}e^{-\gamma (s-\tau)}u(c(s),\bar{L})d s\right| \mathcal{F}_{\tau}\right]$, the gain function of Problem $(P)$ can be rewritten as the expectation of two separated terms representing the pre- and post-retirement part $$J(x;c,\pi,l,\tau)\!=\!\mathbb{E} \left[\!\int_{0}^{\tau}\!e^{\!-\!\gamma t}u(c(t),l(t))d t+e^{\!-\!\gamma \tau}\!J_{\scriptscriptstyle PR}(X(\tau);c,\pi)\right],$$ where the subscript $\scriptstyle PR$ indicates that the corresponding variables and functions are related to the post-retirement problem.

The solutions of the pre- and post-retirement part are based on similar techniques, therefore in this letter we only report the solution of the post-retirement part, referring to the Online Appendix, Section \ref{PR_part}, for details. Depending on the value of $R_{post}$, the solution of the post-retirement problem is divided into two different cases: one is $ R_{post}=\frac{d}{r}$, in which the liquidity constraint has no restriction on the optimization, and the other is $R_{post}> \frac{d}{r}$, with the optimal solution being binded by the liquidity constraint.

	\begin{lemma}\label{Lemma 3}
		The post-retirement value function
\begin{equation*}
U(x)\triangleq\sup_{\{c(t),\pi(t)\}}J_{\scriptscriptstyle PR}(x;c,\pi), 
\end{equation*}
 for $x\ge R_{post}$, is given by:
		\begin{equation*}
			U\left(x\right)=\begin{cases}
				\left(x-\frac{d}{r}\right)^{\delta(1-k)}K_{1}^{1-\delta(1-k)}\bar{L}^{(1-k)(1-\delta)}\frac{1}{\delta(1-k)}, &  \mbox{if}\quad R_{post}=\frac{d}{r},\\
				B_{2,\scriptscriptstyle PR}(\lambda_{\scriptscriptstyle PR}^{*})^{n_{2}}+\frac{1-\delta(1-k)}{\delta(1-k)}K_{1}\bar{L}^{\frac{(1-k)(1-\delta)}{1-\delta(1-k)}}(\lambda_{\scriptscriptstyle PR}^{*})^{\frac{\delta(1-k)}{\delta(1-k)-1}}-\frac{d}{r}\lambda_{\scriptscriptstyle PR}^{*}+\lambda_{\scriptscriptstyle PR}^{*}x,& \mbox{if}\quad R_{post}>\frac{d}{r}.
			\end{cases}
		\end{equation*}
		The Legendre-Fenchel transform of $U(x)$, $\tilde{U}(z)\triangleq\sup\limits_{x\ge R_{post}} [U(x)-zx]$, is:
		\begin{itemize}
			\item  $\tilde{U}(z)=\frac{1-\delta(1-k)}{\delta(1-k)}z^{\frac{\delta(1-k)}{\delta(1-k)-1}}K_{1}\bar{L}^{\frac{(1-k)(1-\delta)}{1-\delta(1-k)}}-\frac{d}{r}z$, $z>0$, if  $R_{post}\!=\!\frac{d}{r}$;
			\item 
				$\tilde{U}(z)\!=\!\begin{cases}
					\!B_{2,\scriptscriptstyle PR}\hat{z}_{\scriptscriptstyle PR}^{n_{2}}\!+\!\frac{1\!-\!\delta(1\!-\!k)}{\delta(1\!-\!k)}K_{1}\!\bar{L}^{\frac{(1\!-\!k)(1\!-\!\delta)}{1\!-\!\delta(1\!-\!k)}}\hat{z}_{\scriptscriptstyle PR}^{\frac{\delta(1\!-\!k)}{\delta(1\!-\!k)\!-\!1}}\!\!\!\!-\!\frac{d}{r}\hat{z}_{\scriptscriptstyle PR}\!\!-R_{post}(z\!-\!\hat{z}_{\scriptscriptstyle PR}), & z\!\ge\!\hat{z}_{\scriptscriptstyle PR},\\
					\!B_{2,\scriptscriptstyle PR}z^{n_{2}}+\frac{1-\delta(1-k)}{\delta(1-k)}K_{1}\bar{L}^{\frac{(1-k)(1-\delta)}{1-\delta(1-k)}}z^{\frac{\delta(1-k)}{\delta(1-k)-1}}-\frac{d}{r}z, & 0\!\!<\!z\!\!<\!\hat{z}_{\scriptscriptstyle PR},
				\end{cases}$ if $R_{post}\!>\!\frac{d}{r}$.
		\end{itemize}
	\end{lemma}

	\begin{proof}
	    See the Online Appendix \ref{PR_part} for the proof and the definition of the constants $\lambda_{\scriptscriptstyle PR}^{*},\,K_1,\,n_2,\,\hat{z}_{\scriptscriptstyle PR}$ and $B_{2,PR}$.
	\end{proof}

	 \subsection{Pre-retirement Part}
Based on the dynamic programming principle, we can only consider a subset of the admissible control set of Problem $(P)$, that is $\mathcal{A}_{1}(x)\subset\mathcal{A}(x)$, in which any policy achieves the maximum of the post-retirement problem's gain function. Hence we have
	\begin{equation*}
		V(x)= \sup_{\left(\tau, \{c(t),\pi(t),l(t) \}\right)\in\mathcal{A}_{1}(x)} \mathbb{E}\left[\int_{0}^{\tau}e^{-\gamma t}u(c(t),l(t))d t+e^{-\gamma\tau}U\left(X^{x,c,\pi,l}(\tau)\right)\right],
	\end{equation*} 
	where $U\left(X^{x,c,\pi,l}(\tau)\right)\!\triangleq\! \!\sup\limits_{ \{c(t),\pi(t),l(t) \}\in\mathcal{A}_{1}(x)}\!\! \mathbb{E}\!\left[\!\left.\int_{\tau}^{\infty}\!\!e^{\!-\!\gamma (s\!-\!\tau)}u(c(s),\bar{L})d s\right|\mathcal{F}_{\tau}\right]$ is given in the previous lemma.

We first define an admissible control set corresponding to a fixed stopping time $\tau\in\mathcal{T}$, with $\mathcal{T}$ representing the set of $\mathcal{F}_{t}$-stopping times, as
	$\mathcal{A}_{\tau}(x)\triangleq \{\{c(t),\pi(t),l(t)\}: \left(\tau,\{c(t),\pi(t),l(t)\}\right)\in\mathcal{A}(x)  \}$,
	and the utility maximization problem
	\begin{equation*}
	    V_{\tau}(x)\triangleq\sup\limits_{\{c(t),\pi(t),l(t)\}\in\mathcal{A}_{\tau}(x) } J(x;c,\pi,l,\tau).   \tag{$P_{\tau}$}
	\end{equation*}
	Then, Problem $(P)$ is converted into an optimal stopping time problem, that is 
$$V(x)=\sup\limits_{\tau\in\mathcal{T}}V_{\tau}(x).$$
	Similar to the post-retirement problem, the primal optimization problem's solution depends on the value of $R_{pre}$, which prompts us to solve it in two different cases. Before the discussion, we follow \cite[Proposition 2.1]{2021arXiv210615426D} to provide the pre-retirement budget constraint, that is:
	\begin{equation}
	    \mathbb{E}\left[\int_{0}^{\tau}H(t)\left(c(t)+d+w l(t)-w\bar{L}\right)d t+H(\tau)X(\tau)\right]\leq x.\label{Equation 6}
	\end{equation}
	Additionally, we define the Legendre-Fenchel transform of $u(c,l)$ by $$\tilde{u}(y)\!\triangleq\!\!\!\sup\limits_{c\ge 0,\,0\leq l\leq L}\!\!\left[u(c,l)\!-\!(c\!+\!w l)y\right].$$

	\subsubsection[Without Liquidity Constraint]{Duality Approach with $R_{pre}=\frac{d-w\bar{L}}{r}$}\label{sub1}
	
	Following the method from \cite[Section 6]{karatzas2000utility}, we first deduce an inequality of $J(x;c,\pi,l,\tau)$  by introducing a Lagrange multiplier $\lambda>0$ and using the budget constraint (\ref{Equation 6}),
	\begin{equation*}
	    J(x;c,\pi,l,\tau)\leq \mathbb{E}\left[\int_{0}^{\tau}e^{-\gamma t}\left(\tilde{u}(\lambda e^{\gamma t}H(t))-(d-w\bar{L})\lambda e^{\gamma t}H(t)\right)d t+e^{-\gamma\tau}\tilde{U}(\lambda e^{\gamma\tau}H(\tau))\right]+\lambda x.
	\end{equation*}
	The inequality turns to equality if and only if the conditions
	\begin{equation*}
		c(t)+w l(t)=-\tilde{u}^{\prime}(\lambda e^{\gamma t}H(t)), \quad \forall t\in[0,\tau], \quad X(\tau)=-\tilde{U}^{\prime}(\lambda e^{\gamma \tau}H(\tau)), \quad \mbox{a.s.,}
	\end{equation*}
	and $\mathbb{E}\!\left[\int_{0}^{\tau}\!\left(c(t)\!+\!w l(t)\!+\!d\!-\!w\bar{L}\right)H(t)d t\!+\!X(\tau)H(\tau)\right]\!=\!x$ hold. 
	
	Additionally, Lemma \ref{Lemma 3} implies
	$X(\tau)\!=\!-\tilde{U}^{\prime}(\lambda e^{\gamma \tau }H(\tau))\!\ge\! R_{post}\!\ge\!\frac{d-w\bar{L}}{r}$.
	Then the following lemma shows that under the above conditions, there always exists a portfolio policy to ensure $X^{x,c,\pi,l}(t)\!\ge\!\frac{d-w\bar{L}}{r}\!=\!R_{pre}$, $\forall t\!\in\![0,\tau]$, which implies the liquidity constraint is satisfied automatically.
	
	\begin{lemma}\label{Lemma 4}
		For any given initial wealth $x\ge R_{pre}$, any fixed stopping time $\tau\in\mathcal{T}$, any $\mathcal{F}_{\tau}$-measurable random variable $K$ with $\mathbb{P}(K\ge\frac{d-w\bar{L}}{r})=1$ under the $\mathbb{P}$ measure, and any given progressively measurable consumption and leisure processes $c(t)\ge 0$, $l(t)\ge 0$, $\forall t\ge 0$, satisfying
		$\mathbb{E}\left[\int_{0}^{\tau}H(t)(c(t)+w l(t)+d-w\bar{L})d t+H(\tau )K\right]=x$,
		there exists a portfolio process $\pi(t)$ making $X^{x,c,\pi,l}(t)\ge \frac{d-w\bar{L}}{r}$, $\forall t\in[0,\tau]$, and $X^{x,c,\pi,l}(\tau)=K$ hold almost surely.
	\end{lemma}
	
	\begin{proof}
	    See Online Appendix \ref{Appendix 4}.
	\end{proof}
	
	Moreover, the Lagrange method indicates that $V_{\tau}(x)=\inf\limits_{\lambda>0} [\tilde{J}_{\tau}(\lambda)+\lambda x]$ with introducing
	\begin{equation*}
		\tilde{J}_{\tau}(\lambda)\triangleq\mathbb{E}\left[\int_{0}^{\tau}e^{-\gamma t}\left(\tilde{u}(\lambda e^{\gamma t}H(t))-(d-w\bar{L})\lambda e^{\gamma t}H(t)\right)d t+e^{-\gamma\tau}\tilde{U}(\lambda e^{\gamma\tau}H(\tau))\right],
	\end{equation*}
	and the value function of Problem $(P)$ can be transformed as
	\begin{equation*}
		V(x)=\sup_{\tau \in \mathcal{T}}V_{\tau}(x)=\sup_{\tau \in \mathcal{T}}\inf_{\lambda>0}[\tilde{J}_{\tau}(\lambda)+\lambda x]\leq \inf_{\lambda>0}\sup_{\tau \in \mathcal{T}}[\tilde{J}_{\tau}(\lambda)+\lambda x]=\inf_{\lambda>0}[\sup_{\tau \in \mathcal{T}}\tilde{J}_{\tau}(\lambda)+\lambda x].
	\end{equation*}
	Defining $\tilde{V}(\lambda)\triangleq\sup\limits_{\tau \in \mathcal{T}}\tilde{J}_{\tau}(\lambda)$, \cite[Section 8, Theorem 8.5]{karatzas2000utility} shows that $V(x)=\inf\limits_{\lambda>0}[\tilde{V}(\lambda)+\lambda x]$
	holds under the condition that the function $\tilde{V}(\lambda)$ exists and it is differentiable for any $\lambda>0$. Then, the process to solve Problem $(P)$ is divided into two steps: the first is involved in the pure optimal stopping time problem of $\tilde{V}(\lambda)$, and the second step mainly concerns finding the optimal Lagrange multiplier. We begin with the first step and introduce a new process, $Z(t)\triangleq\lambda e^{\gamma t}H(t)$. Then $\tilde{V}(\lambda)$ can be rewritten as
	\begin{equation*}
		\tilde{V}(\lambda)=\sup_{\tau \in S}\mathbb{E}\left[\int_{0}^{\tau}e^{-\gamma t}\left(\tilde{u}(Z(t))-(d-w\bar{L})Z(t)\right)d t+e^{-\gamma\tau}\tilde{U}(Z(\tau))\right].
	\end{equation*}
	We proceed with a generalized optimal stopping time problem
	\begin{equation}
		\phi(t,z)=\sup_{\tau\ge t}\mathbb{E}\left[\left.\int_{t}^{\tau}e^{-\gamma s}\left(\tilde{u}(Z(s))-(d-w\bar{L})Z(s)\right)d s+e^{-\gamma \tau}\tilde{U}(Z(\tau))\right|Z(t)=z\right],\label{Equation 8}
	\end{equation}
	which shows that $\tilde{V}(\lambda)=\phi(0,\lambda)$. The following lemma gives the continuous region and stopping region of the above optimal stopping time problem.
	\begin{lemma}\label{Lemma 5}
		Considering the optimal stopping time problem (\ref{Equation 8}) with the state variable $Z(t)$, the continuous region is $\Omega_{1}\!=\!\{Z(t)\!>\!\bar{z}\}$, the stopping region is $\Omega_{2}\!=\!\{0\!<\!Z(t)\!\leq\!\bar{z}\}$, where $\bar{z}$ denotes the boundary that separates $\Omega_{1}$ and $\Omega_{2}$. 
	\end{lemma}
	
	\begin{proof}
	    See Online Appendix \ref{Appendix 5}.
	\end{proof}
	
%
%
	
	Straight follows, with the operator $\mathcal{L}=\frac{\partial}{\partial t}+(\gamma-r)z\frac{\partial}{\partial z}+\frac{\theta^{2}}{2}z^{2}\frac{\partial^{2}}{\partial z^{2}}$,
	the optimal stopping time problem (\ref{Equation 8}) is equivalent to solving the free boundary problem below.\\
	\textbf{Variational Inequalities:} Find a free boundary $\bar{z}\!>\!0$ (Retirement level), and a function $\phi(t,z)\!\in\! C^{1}\!\left((0,\infty)\!\times\!\mathbb{R}^{+}
	\right)\cap C^{2}\left((0,\infty)\times\left(\mathbb{R}^{+}\setminus \{\bar{z}\}\right)\right)$ satisfying
	\begin{equation}
	\begin{cases}
	(V1) \quad\mathcal{L}\phi(t,z)+e^{-\gamma t}\left(\tilde{u}(z)-(d-w\bar{L})z\right)=0, & z> \bar{z},\\
	(V2) \quad\mathcal{L}\phi(t,z)+e^{-\gamma t}\left(\tilde{u}(z)-(d-w\bar{L})z\right)\leq 0, &0<z\leq \bar{z},\\
	(V3) \quad\phi(t,z)\ge e^{-\gamma t}\tilde{U}(z), &z>\bar{z}, \\
	(V4) \quad\phi(t,z)=e^{-\gamma t}\tilde{U}(z), &0<z\leq\bar{z},\label{Equation 9} 
	\end{cases}
	\end{equation}
	for any $t\ge 0$, with the smooth fit conditions $\phi(t,\bar{z})=e^{-\gamma t}\tilde{U}(\bar{z})$ and $\frac{\partial \phi}{\partial z}(t,\bar{z})=e^{-\gamma t}\tilde{U}^{\prime}(\bar{z})$.
	The analytical solution of the above inequalities is presented in Online Appendix \ref{Appendix 7}.
	
Once $\phi$ is computed, we recover $\tilde{V}(\lambda)=\phi(0,\lambda)$, and the value function is given by
\begin{equation*}
		V(x)=\inf_{\lambda>0}[\tilde{V}(\lambda)+\lambda x]=\tilde{V}(\lambda^*)+\lambda^* x,
	\end{equation*}
$x$ being the initial wealth. The retirement time is the first time the process $Z^*(t)\triangleq\lambda^* e^{\gamma t}H(t)$ touches the barrier $\bar{z}$ from above. The optimal strategies are reported at the end of the Online Appendix \ref{Appendix 7}.

\begin{remark}\label{Rem1}
The optimal retirement time is the first time the process $Z^*(t)$ touches the lower barrier $\bar{z}$. The same can be obtained with respect to the wealth level $X(t)$. In fact, the optimal process $Z^*$ is connected to the optimal wealth $X$ by the relation $X(t)=-v^{\prime}(Z^*(t))$, being ${\phi}(t,z)=e^{-\gamma t}v(z),$ see the online appendix. The convex property of $v(\cdot)$, see \cite[Section 3.4, Lemma 4.3]{karatzas1998methods}, indicates that $X(t)$ is a decreasing function of $Z^*(t)$, therefore, in this case the optimal retirement time is the first time the process $X(t)$ touches an upper barrier $\bar{x}=-v^{\prime}(\bar{z})$.
\end{remark}

	\subsubsection[With Liquidity Constraint]{Duality Approach with $R_{pre}>\frac{d-w\bar{L}}{r}$}
	
	Before proceeding to solve the problem, we present the following proposition to construct expectation form of the liquidity constraint related to $X(t)\!\ge\! R_{pre}$, $\forall t\in[0,\tau]$.
	\begin{proposition}
		The liquidity constraint of the considered problem is
		\begin{equation}
			\mathbb{E}\left[\left.\int_{t}^{\tau}\frac{H(s)}{H(t)}\left(c(s)+d+w l(s)-w\bar{L}\right)d s+\frac{H(\tau)}{H(t)}X(\tau)\right|\mathcal{F}_{t}\right]\ge R_{pre}, \quad\forall t\in[0,\tau].\label{Equation 10}
		\end{equation}
	\end{proposition}
	\begin{proof}
	    See \cite[Proposition 4.1]{2021arXiv210615426D}.
	\end{proof}
	
	Considering the budget and liquidity constraints, (\ref{Equation 6}) and (\ref{Equation 10}), and introducing a Lagrange multiplier $\lambda>0$ and a non-increasing process $D(t)\ge 0$  \cite{karatzas2000utility,he1993labor}, the following inequality is obtained:
	\begin{equation*}
	\begin{split}
	    J(x;c,\pi,l,\tau)
	    &\!\leq\!\mathbb{E}\!\left[\!\int_{0}^{\tau}\!e^{\!-\!\gamma t}\left(\tilde{u}(\lambda D(t)e^{\gamma t}H(t))\!-\!(d\!-\!w\bar{L})\lambda e^{\gamma t}D(t)H(t)\right)\!d t\!+\!e^{\!-\!\gamma\tau}\tilde{U}(\lambda D(\tau)e^{\gamma\tau}\!H(\tau))\!\right]\\
	    &\qquad
	    +\lambda\mathbb{E}\left[\int_{0}^{\tau}R_{pre}H(t)d D(t)\right]+\lambda x,
	\end{split}
	\end{equation*}
	which inspires us to define the dual individual's shadow price problem
	\begin{equation*}\tag{$S_{\tau}$}
	\begin{split}
	    \tilde{V}_{\tau}(\lambda)
	    &\triangleq \inf_{D(t)\in\mathcal{D}}\mathbb{E}\bigg[\int_{0}^{\tau}e^{-\gamma t}\left(\tilde{u}(\lambda D(t)e^{\gamma t}H(t))-(d-w\bar{L})\lambda e^{\gamma t}D(t)H(t)\right)d t\\
	   &\qquad+e^{-\gamma\tau}\tilde{U}(\lambda D(\tau)e^{\gamma\tau}H(\tau))\bigg]
	   +\lambda\mathbb{E}\left[\int_{0}^{\tau}R_{pre}H(t)d D(t)\right],
	\end{split} 
	\end{equation*}
	where $\mathcal{D}$ is the set of non-negative, non-increasing and progressively measurable processes. Then we establish the duality between Problem $(S_{\tau})$ and $(P_{\tau})$.
	
	\begin{theorem}
	\textbf{(Duality Theorem)}
	Suppose $D^{*}(t)$ is the optimal solution to Problem $(S_{\tau})$, then $c^{*}(t)\!+\!w l^{*}(t)\!=\!-\!\tilde{u}^{\prime}(Z^*(t))$ and $X^{x,c^{*},\pi^{*},l^{*}}(\tau)\!=\!-\tilde{U}^{\prime}(Z^*(\tau))$ coincide with the optimal solution of Problem $(P_{\tau})$, and there exists
	$V_{\tau}(x)\!=\!\inf\limits_{\lambda>0}\! \left[\!\tilde{V}_{\tau}(\lambda)\!+\!\lambda x\right]$, $\forall x\!\ge\! R_{pre}$. Here $Z^*(t)=\lambda^{*} e^{\gamma t}D^{*}(t)H(t)$, where $\lambda^{*}$ and $D^*(t)$ are the parameters $\lambda$ and $D(t)$ giving the infimum.
	\end{theorem}
	
	\begin{proof}
	    See \cite[Theorem 4.1]{2021arXiv210615426D}.
	\end{proof}
	
	\noindent This duality theorem allows us to link Problem $(P)$ with the shadow price problem through
	\begin{equation*}
		V(x)=\sup_{\tau\in\mathcal{T}}V_{\tau}(x)=\sup_{\tau\in\mathcal{T}}\inf_{\lambda>0}[\tilde{V}_{\tau}(\lambda)+\lambda x]\leq\inf_{\lambda>0}\sup_{\tau\in\mathcal{T}}[\tilde{V}_{\tau}(\lambda)+\lambda x]=\inf_{\lambda>0}[\sup_{\tau\in\mathcal{T}}\tilde{V}_{\tau}(\lambda)+\lambda x].
	\end{equation*}
	Defining
	$\tilde{V}(\lambda)\triangleq\sup\limits_{\tau\in\mathcal{T}}\tilde{V}_{\tau}(\lambda)$, \cite[Section 8, Theorem 8.5]{karatzas2000utility} indicates that the last inequality takes the equal sign with the condition that $\tilde{V}(\lambda)$ exists and is differentiable for any $\lambda> 0$. Thereafter, the objective optimization problem can be divided into two parts:
	\begin{equation*}
	\begin{cases}
	    \tilde{V}(\lambda)=\sup\limits_{\tau \in S}\tilde{V}_{\tau}(\lambda),\\
	    V(x)=\inf\limits_{\lambda>0} [\tilde{V}(\lambda)+\lambda x]\triangleq\tilde{V}(\lambda^{*})+\lambda^{*} x.
	\end{cases}
	\end{equation*}
	We now consider the technique of \cite{davis1990portfolio} and insert an assumption on the process $D(t)$ for acquiring a closed-form solution.
	\begin{assumption}
		The non-increasing process $D(t)$ is absolutely continuous with respect to $t$. Hence, there is a non-negative process $\psi(t)$ such that $d D(t)=-\psi(t)D(t)d t$.
	\end{assumption}

	Then, by means of a new defined process $Z(t)\triangleq\lambda D(t)e^{\gamma t}H(t)$, the value function of the individual's shadow price problem can be written as
	\begin{equation*}
		\tilde{V}_{\tau}(\lambda)=\inf_{\psi(t)\ge 0}\mathbb{E}\left[\int_{0}^{\tau} e^{-\gamma t} \left(\tilde{u}(Z(t))-(d-w\bar{L})Z(t)-R_{pre}\psi(t)Z(t)\right)d t+e^{-\gamma \tau}\tilde{U}(Z(\tau))\right],
	\end{equation*}
	where $\psi(t)$ is the control variable, and $Z(t)$ is the state variable. Introducing a generalized problem
	\begin{equation*}
		\phi(t,z)\!\triangleq\!\sup_{\tau\ge t}\! \inf_{\psi(t)\ge 0}\! \mathbb{E}\!\left[\!\left.\int_{t}^{\tau}\!\!\!\!\!e^{\!-\!\gamma s}\!\! \left(\tilde{u}(Z(s))\!-\!(d\!-\!w\bar{L})Z(s)\!-\!R_{pre}\psi(s)Z(s)\right)d s\!+\!e^{\!-\!\gamma \tau}\!\tilde{U}(Z(\tau))\right|\!Z(t)\!=\!z\!\right],
	\end{equation*}
	the solution of $\tilde{V}(\lambda)$ is turned to $\phi(t,z)$ with $\tilde{V}(\lambda)\!=\!\phi(0,\lambda)$. We first handle the infimum part by defining
	\begin{equation*}
		\phi_{\scriptscriptstyle inf}(t,z)\!\triangleq\!\inf_{\psi(t)>0} \mathbb{E}\!\left[\!\left.\int_{t}^{\tau} \!\!\!\!\!e^{\!-\!\gamma s}\! \left(\tilde{u}(Z(s))\!-\!(d\!-\!w\bar{L})Z(s)\!-\!R_{pre}\psi(s)Z(s)\right)\!d s\!+\!e^{\!-\!\gamma \tau}\!\tilde{U}(Z(\tau))\right|\!Z(t)\!=\!z\right].
	\end{equation*}
	The corresponding Bellman equation is
	\begin{equation*}
		\min_{\psi\ge 0} \left\{\mathcal{L}\phi_{\scriptscriptstyle inf}(t,z)+e^{-\gamma t}\left(\tilde{u}(z)-(d-w\bar{L})z\right)-\psi z\left[\frac{\partial \phi_{\scriptscriptstyle inf}}{\partial z}(t,z)+R_{pre}e^{-\gamma t}\right]\right\}=0.
	\end{equation*}
	The optimum $\psi^{*}$ has the following characterization,
	\begin{itemize}
	\item \!$\frac{\partial \phi_{\scriptscriptstyle inf}}{\partial z}(t,z)\!+\!R_{pre}e^{\!-\!\gamma t}\!=\!0 \Rightarrow \psi^{*}\!\!\ge\! 0$ and $\frac{\partial \phi}{\partial z}(t,z)=\frac{\partial \phi_{\scriptscriptstyle inf}}{\partial z}(t,z)=-R_{pre}e^{-\gamma t},\quad z\ge\hat{z}$.
	\item \!$\frac{\partial \phi_{\scriptscriptstyle inf}}{\partial z}(t,z)\!+\!R_{pre}e^{\!-\!\gamma t}\!\leq \!0 \Rightarrow \psi^{*}\!\!=\!0$, then $\phi(t,z)$ switches to a pure optimal stopping time problem,
	\begin{equation*}
	    \phi(t,z)\!=\! \sup_{\tau\ge t} \mathbb{E}\!\left[\!\left.\int_{t}^{\tau}\! \!e^{\!-\!\gamma s} \left(\tilde{u}(Z(s))\!-\!(d\!-\!w\bar{L})Z(s)\right)d s\!+\!e^{\!-\!\gamma \tau}\tilde{U}(Z(\tau))\right|Z(t)\!=\!z\right],
	\end{equation*} 
	which has the same form as (\ref{Equation 8}) but applies to the interval $0<z<\hat{z}$.
	\end{itemize}
Lemma \ref{Lemma 5} can be easily extended also in this case, therefore the optimal retirement time is the first time the process $Z^*(t)$ touches the lower barrier $\bar{z}$. Therefore, we need to compare the value of $\bar{z}$ and $\hat{z}$, and split the discussion into two cases: the first one is $\bar{z}<\hat{z}$, which corresponds to the case where the liquidity constraint boundary, $R_{pre}$,  is lower than the retirement threshold. 

	\noindent \textbf{Variational Inequalities assuming $\bar{z}<\hat{z}$:} 
	Find the free boundaries $\bar{z}>0$ (retirement), $\hat{z}>0$ ($R_{pre}$-wealth level), and a function ${\phi}(\cdot,\cdot)\in C^{1}((0,\infty)\times\mathbb{R}^{+})\cap C^{2}((0,\infty)\times\mathbb{R}^{+}\setminus\{\bar{z}\})$ satisfying
	\begin{equation}
		\begin{cases}
			
			(V1) \quad\frac{\partial {\phi}}{\partial z}(t,z)+R_{pre}e^{-\gamma t}=0, & z\ge\hat{z},\\
			(V2) \quad \frac{\partial {\phi}}{\partial z}(t,z)+R_{pre}e^{-\gamma t}\leq 0, & 0<z<\hat{z},\\
			(V3)\quad \mathcal{L}{\phi}(t,z)+e^{-\gamma t}\left(\tilde{u}(z)-(d-w\bar{L})z\right)=0, & \bar{z}<z<\hat{z},\\
			(V4)\quad \mathcal{L}{\phi}(t,z)+e^{-\gamma t}\left(\tilde{u}(z)-(d-w\bar{L})z\right)\leq 0, & 0<z\leq\bar{z},\\
			(V5) \quad {\phi}(t,z)\ge e^{-\gamma t}\tilde{U}(z), & \bar{z}<z<\hat{z},\\
			(V6) \quad {\phi}(t,z)= e^{-\gamma t}\tilde{U}(z), & 0<z\leq\bar{z},\label{Equation 12}
		\end{cases}
	\end{equation}
	for any $t\ge 0$, with the smooth fit conditions
	\begin{equation*}
	   \frac{\partial {\phi}}{\partial z}(t,\hat{z})=-R_{pre}e^{-\gamma t},\quad \frac{\partial^{2}{\phi}}{\partial z^{2}}(t,\hat{z})=0, \quad {\phi}(t,\bar{z})=e^{-\gamma t}\tilde{U}(\bar{z}),\quad \mbox{and} \quad \frac{\partial {\phi}}{\partial z}(t,\bar{z})=e^{-\gamma t}\tilde{U}^{\prime}(\bar{z}). 
	\end{equation*}
	The analytical solution of the variational equation (\ref{Equation 12}) is reported in Online Appendix \ref{Appendix 8}. Once $\phi$ and $\bar{z}$ are computed, the value function and the optimal retirement decision can be recovered as in Section \ref{sub1}. The optimal strategies are reported at the end of the Online Appendix \ref{Appendix 8}.

If the first case does not admit a solution, that is, the liquidity constraint boundary $R_{pre}$ is high enough (and larger than $R_{post}$) to make the agent declare retirement at time 0 for any admissible initial wealth, we deal with an immediate retirement, and therefore $V(x)=U(x)$, and all the optimal strategies are the ones of the post-retirement problem. 

	\section{Numerical Analysis}
	 
	 We now perform the sensitivity analysis to the liquidity constraint boundaries. All the input parameters are reported in Table \ref{Table 1}. We change the values of $R_{pre}$, $R_{post}$ and keep all other input parameters consistent with Table \ref{Table 1} to discover the different convergence phenomena of retirement wealth threshold concerning the pre- and post-retirement liquidity constraints.
	\begin{table}[t]
	    \centering
		\renewcommand\arraystretch{1.2}
		\caption{Input Parameters}\label{Table 1}
		\begin{tabular}{c c c c c c c c c c c c } \hline
			\textbf{ $\delta$} & \textbf{ $k$ } & \textbf{ $r$ } & \textbf{$\mu$} &\textbf{$\sigma$} & \textbf{ $\gamma$ } & \textbf{ $d$ } & \textbf{ $w$ } & \textbf{ $R_{pre}$ }  & \textbf{ $R_{post}$ }& \textbf{ $\bar{L}$ } & \textbf{ $L$ } \\ \hline
			0.6  & 3 & 0.02 & 0.07 & 0.15 & 0.1 & 0.3 & 1.5 & 0 & 15 & 1 & 0.8\\ \hline
		\end{tabular}
	\end{table}
	
	   Figure \ref{Figure 2} shows that the retirement wealth threshold $\bar{x}$ is a decreasing function of $R_{pre}$ due to the fact that the agent with higher $R_{pre}$ values prefers to set a lower wealth threshold to make sure entering in retirement ahead of schedule such that getting rid of the restriction caused by $R_{pre}$. Whereas, the critical wealth level of retirement is increasing with respect to $R_{post}$. Since the pre-retirement restriction keeps constant, a higher value of $R_{post}$, which implies a more rigorous circumstance for the post-retirement period, impels the agent to step into retirement with a higher wealth level.
	 
	 \begin{figure}[t]
		\centering
		\caption{Convergence w.r.t. Liquidity Constraint Boundary of Pre- and Post-Retirement Part}\label{Figure 2}
		\includegraphics[width=1\linewidth,height=0.3\textheight]{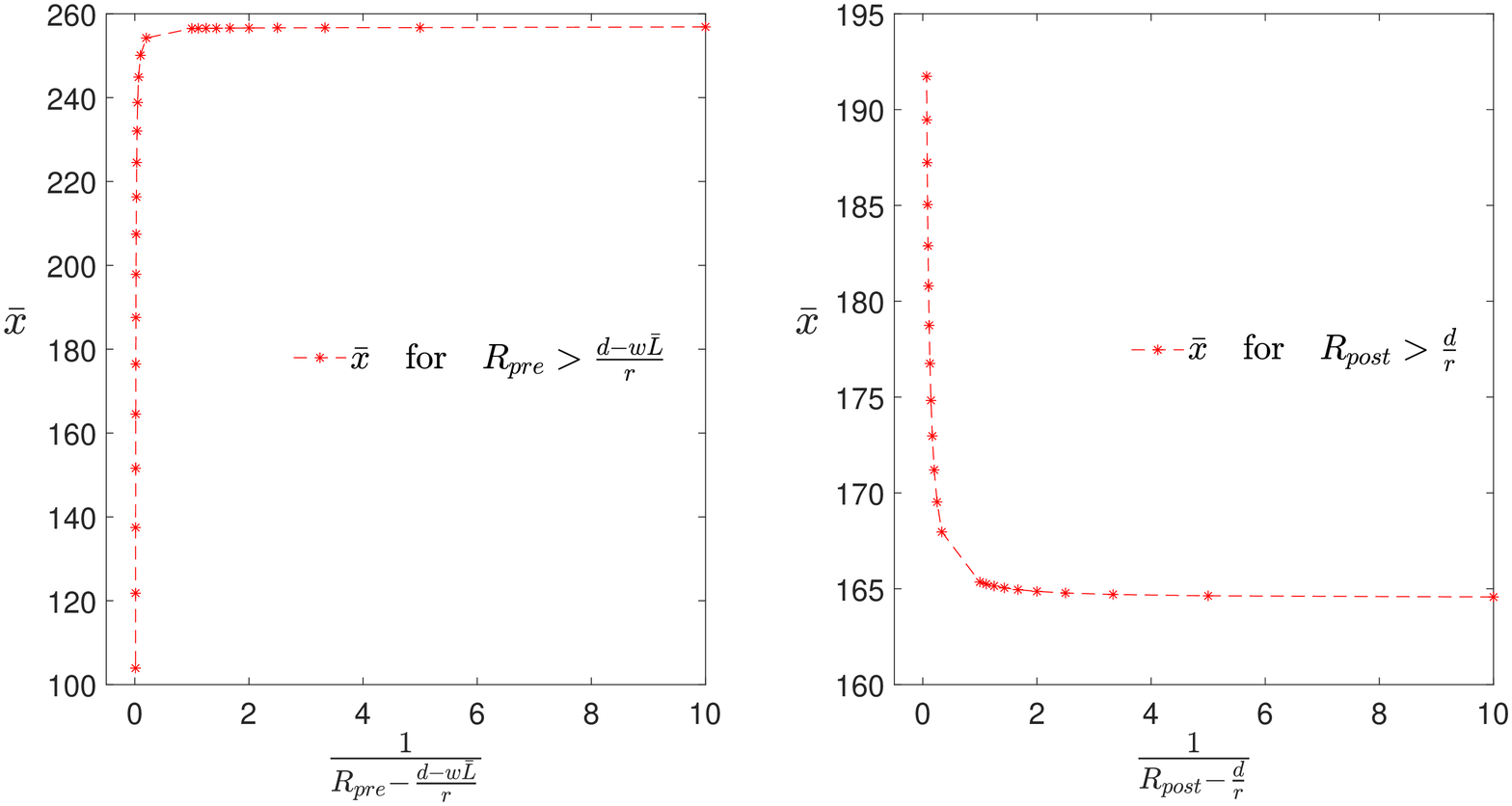}
	 \end{figure}
	 
	 Moreover, we provide figures to illustrate the sensitivity of optimal consumption, portfolio and leisure fractions in terms of $x-\frac{d}{r}$ with respect to different values of $R_{pre}$ and $R_{post}$. We begin this kind of analysis fixing the value of $R_{post}$ and arranging three values to $R_{pre}$. The optimal control strategies for different cases are presented in Proposition \ref{Proposition 1}, Proposition \ref{Proposition Appendix 5} and Proposition \ref{Proposition Appendix 6} in the online appendix.
	\begin{figure}[t]
		\centering
		\caption{Optimal Control Fractions w.r.t. Liquidity Constraint Boundary of Pre-Retirement Part}\label{Figure 5}
		\includegraphics[width=1\linewidth,height=0.35\textheight]{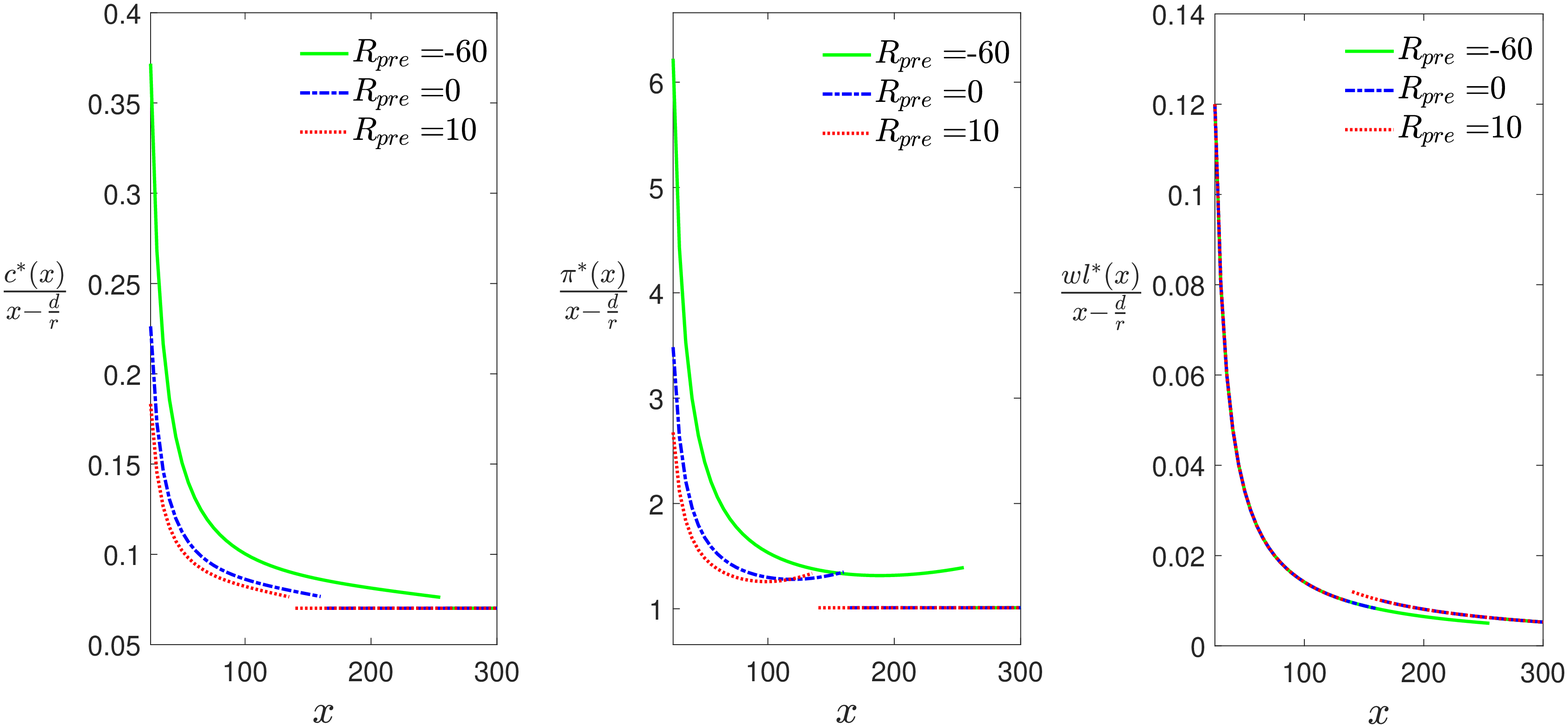}
	\end{figure}
	\noindent In Figure \ref{Figure 5}, $R_{post}$ is set equal to $\frac{d}{r}=15$, which implies that the post-retirement part is not restricted by the liquidity constraint. We can observe that the optimal consumption and portfolio fractions suffer a downward jump for various $R_{pre}$ values. This is due to the discontinuity of the leisure rate at the retirement time, which leads to a shrinkage of labour income and reduces the resources allocated to the consumption and investment. In fact, if $x>256.6913$ $(R_{pre}=-60)$, $x>164.5320$ $(R_{pre}=0)$, $x>137.4776$ $(R_{pre}=10)$, i.e., the initial wealth is larger than the retirement threshold $\bar{x}$, the agent is facing the post-retirement region, with $l^{*}(t)=\bar{L}=1$ (full leisure). In addition, it should be noted that for different $R_{pre}$ values, the jump happens at different wealth levels. As also shown in the left plot of Figure \ref{Figure 2}, the agent with a higher $R_{pre}$ value experiences the jump at a lower wealth threshold of retirement $\bar{x}$. Moreover, since the value of $R_{post}$ keeps identical, the optimal consumption and portfolio fractions of different curves are coincident for the post-retirement part and equal to a constant, in line with the Merton classical problem.
	
	Then we conduct a similar sensitivity analysis with respect to $R_{post}$. Figure \ref{Figure 6} shows that the retirement threshold is increasing with the value of $R_{post}$ ($\bar{x}=164.5320$ for $R_{post}=15$, $\bar{x}=171.1993$ for $R_{post}=20$, $\bar{x}=180.7943$ for $R_{post}=25$), in line with the right plot of Figure \ref{Figure 2}, and describes that the optimal control fractions for the post-retirement part of blue dashed and red dotted curves, whose $R_{post}$ values are greater than the boundary $\frac{d}{r}$, i.e., the liquidity constraints impose restrictions on optimal solutions, converge to the ones of the green curve ( $R_{post}=\frac{d}{r}$) as $x$ increases. It can be explained by the fact that the liquidity constraint plays a slighter role as the wealth becomes comparably larger and imposes a weaker restriction on the admissible control set. Moreover, we also notice that a high liquidity constraint for the post-retirement part induces the agent to take a large risk (high value of $\pi^*$) when the retirement threshold is close. 	
	\begin{figure}[t]
		\centering
		\caption{Optimal Control Fractions w.r.t. Liquidity Constraint Boundary of Post-Retirement Part}\label{Figure 6}
		\includegraphics[width=1\linewidth,height=0.35\textheight]{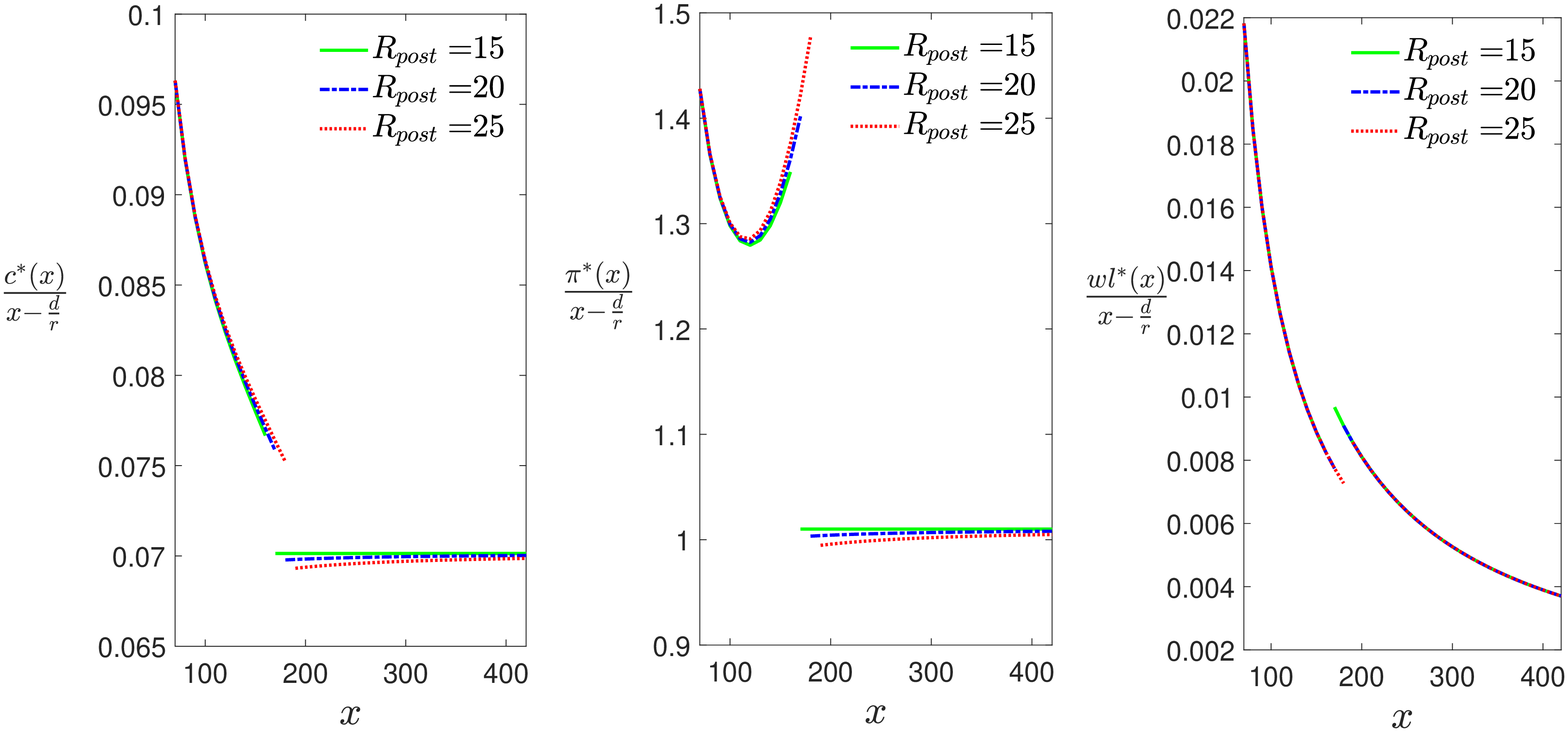}
	\end{figure} 
	
	 Finally, we conduct the sensitivity analysis of optimal control strategies to both the liquidity constraint boundary and the retirement option. In Figure \ref{Figure 3}, we fix the value of $R_{post}$ to $\frac{d}{r}$ and plot the curves of optimal consumption and portfolio fractions in terms of $x-\frac{d}{r}$ under different situations. The dashed lines represent the optimal control fractions of different $R_{pre}$ values with retirement option, while the solid lines represent the corresponding optimal control fractions without retirement option (and therefore with fixed liquidity constraint $R=R_{pre}$).\footnote{The theoretical solutions of optimal consumption-portfolio problem without retirement comes from \cite[Section 5]{2021arXiv210615426D} by replacing the liquidity constraint boundary $F+\eta$ with $R_{pre}$.}  From all the dashed lines, we can see that the optimal consumption and portfolio fractions suffer a downward jump for various $R_{pre}$ values. This is due to the discontinuity of the leisure rate at the retirement time, which leads to a shrinkage of labour income and reduces the resources allocated to the consumption and investment. Comparing the solid and dashed lines with the same colour, the agent with the additional retirement option tends to consume less and invest more in the risky asset for the motivation of arriving at the retirement wealth threshold and enjoying the full leisure rate faster. This kind of difference becomes more significant as the wealth approaches the critical level. Furthermore, the degree of this motivation is related to the liquidity constraint boundary. Observing the convexity of the pre-retirement part of different dashed lines, the optimal control fraction with a higher $R_{pre}$ value takes a larger convexity, which is because stricter liquidity constraints give the agent a stronger motivation to achieve the critical wealth level to get rid of this restriction.

	\begin{figure}[tb]
		\centering
		\caption{Optimal Control Fractions w.r.t. Liquidity Constraint Boundary and Retirement Option.}\label{Figure 3}
		\includegraphics[width=1\linewidth,height=0.3\textheight]{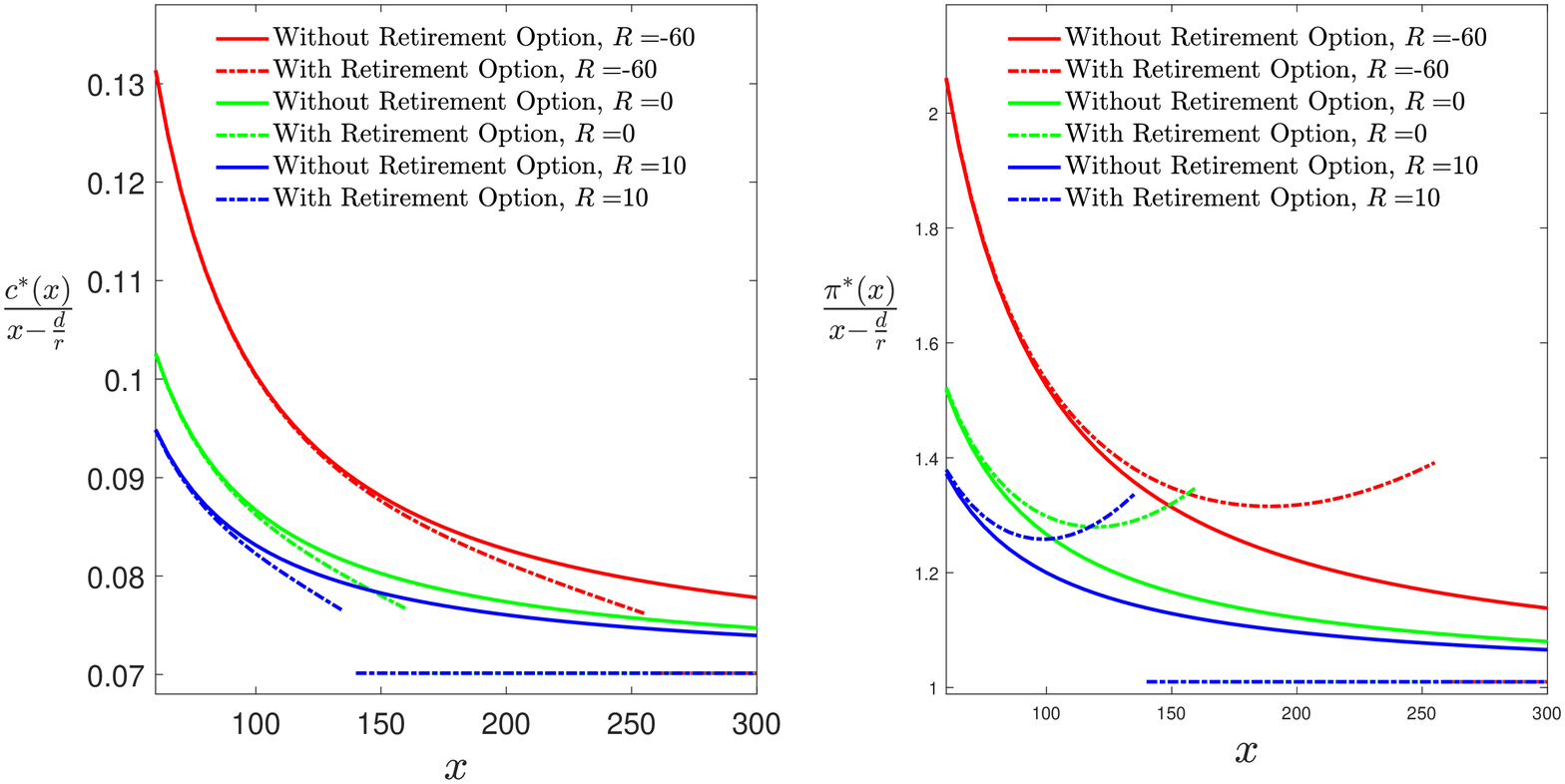}
	 \end{figure}

	 \bibliographystyle{ieeetr}
	 \bibliography{ReferenceOfP2}
	\clearpage
	 \newpage
	 \setcounter{page}{1}
	 \setcounter{footnote}{0}
\begin{center}
	 \title{\bf \Large Sensitivity of Optimal Retirement Problem to Liquidity Constraints - Online Appendix }
	 
	\author{\sffamily Guodong Ding$^{1}$, Daniele Marazzina$^{1,2}$\\
		{\sffamily\small $^1$ Department of Mathematics, Politecnico di Milano}\\{\sffamily \small Piazza Leonardo da Vinci 32, I-20133, Milano, Italy}\\
		{\sffamily\small $^2$ Corresponding Author, daniele.marazzina@polimi.it}}
	 \end{center}
	 
	 \appendix 

	\section{Post Retirement Part}\label{PR_part}
	
	 Assuming $\tau=0^{-}$, we deal with the post-retirement problem, which is an infinite-time optimization problem with two control variables, the consumption and portfolio processes. Introducing $u_{\scriptscriptstyle PR}(c)\triangleq u(c,\bar{L})=\frac{c^{\delta(1-k)}\bar{L}^{(1-\delta)(1-k)}}{\delta(1-k)}$, the corresponding value function, denoted as $(P_{\scriptscriptstyle PR})$, is
	\begin{equation*}
	    V_{\scriptscriptstyle PR}(x)\triangleq\sup_{\{c(t),\pi(t)\}\in\mathcal{A}_{\scriptscriptstyle PR}(x)}J_{\scriptscriptstyle PR}(x;c,\pi). \tag{$P_{\scriptscriptstyle PR}$}
	\end{equation*}
	 
	 \noindent The admissible control set $\mathcal{A}_{\scriptscriptstyle PR}(x)$ takes the compatible definition with $\mathcal{A}(x)$, except that the condition for stopping time is abolished, and the condition for liquidity constraint is given by $X(t)\!\ge\! R_{post}$, a.s., $\forall t\!\ge\! 0$. Then we derive the derivative function
	 $u_{\scriptscriptstyle PR}^{\prime}(c)=c^{\delta(1-k)-1}\bar{L}^{(1-k)(1-\delta)}$,
	 which is positive and strictly decreasing and has the inverse function 
	 $I_{\scriptscriptstyle PR}(z)\triangleq z^{\frac{1}{\delta(1-k)-1}}\bar{L}^{\frac{(1-k)(1-\delta)}{1-\delta(1-k)}}$. Furthermore, referring to \cite[Section 3, Definition 4.2]{karatzas1998methods}, we define the Legendre-Fenchel transform of $u_{\scriptscriptstyle PR}(z)$ as
	 $\tilde{u}_{\scriptscriptstyle PR}(z)\triangleq\sup\limits_{c\ge 0}\left[u_{\scriptscriptstyle PR}(c)-c z\right]$, which has the explicit expression
	 \begin{equation*}
	     \tilde{u}_{\scriptscriptstyle PR}(z)=u_{\scriptscriptstyle PR}(I_{\scriptscriptstyle PR}(z))-z I_{\scriptscriptstyle PR}(z)=\frac{1-\delta(1-k)}{\delta(1-k)}z^{\frac{\delta(1-k)}{\delta(1-k)-1}}\bar{L}^{\frac{(1-k)(1-\delta)}{1-\delta(1-k)}}.
	 \end{equation*}
	 
	 \begin{proposition}
	 The infinite horizon budget constraint of the post-retirement problem is
	 \begin{equation}
	     \mathbb{E}\left[\int_{0}^{\infty}H(t)(c(t)+d)d t\right]\leq x.\label{Equation 1}
	\end{equation}
	\end{proposition}
	\begin{proof}
	The proof can be accomplished directly by replacing $l(t)$ in \cite[Proposition 3.1]{2021arXiv210615426D} with the constant $\bar{L}$, meanwhile inserting a constant term $d$ in the integral.
	\end{proof}
	
	 Depending on the value of $R_{post}$, the solution of Problem $(P_{\scriptscriptstyle PR})$ is divided into two different cases. One is $ R_{post}=\frac{d}{r}$, in which the liquidity constraint has no restriction on the optimization, and the other is $R_{post}> \frac{d}{r}$, with the optimal solution being binded by the liquidity constraint.
	
	 As in \cite[Chapter 3, Example 9.22]{karatzas1998methods}, the optimal wealth process under the condition $ R_{post}\!=\!\frac{d}{r}$ is
	 $X^{*}(t)\!=\!\left(x\!-\!\frac{d}{r}\right)e^{\frac{1}{1\!-\!\delta(1\!-\!k)}\left(r\!-\!\gamma\!+\!\frac{\theta^{2}}{2}\right)t+\frac{\theta}{1\!-\!\delta(1\!-\!k)} B(t)}\!+\!\frac{d}{r}$. The optimal consumption-portfolio polices are $c^{*}(t)\!=\!\frac{1}{K_{1}}\!\!\left[\!X^{*}(t)\!-\!\frac{d}{r}\!\right]$ and $\pi^{*}(t)\!=\!\frac{\theta}{\sigma(1\!-\!\delta(1\!-\!k))}\!\!\left[\!X^{*}(t)\!-\!\frac{d}{r}\!\right]$, with $K_{1}\triangleq\frac{1-\delta(1-k)}{\gamma-r\delta(1-k)-\frac{\theta^{2}}{2}\frac{\delta(1-k)}{1-\delta(1-k)}}>0$. And the value function of Problem $(P_{\scriptscriptstyle PR})$ can be obtained as
	\begin{equation}
	    V_{\scriptscriptstyle PR}(x)=\left(x-\frac{d}{r}\right)^{\delta(1-k)}K_{1}^{1-\delta(1-k)}\frac{\bar{L}^{(1-k)(1-\delta)}}{\delta(1-k)}.\label{Equation 5}
	\end{equation}
	\begin{remark}
	Analogous to the solution of the Merton problem, under the infinite time horizon, the optimal fraction invested in the risky asset in terms of the wealth minus the debt, i.e., $\frac{\pi^{*}(t)}{X^{*}(t)-\frac{d}{r}}$ keeps constant as $-\frac{\theta}{\sigma(\delta(1-k)-1)}=\frac{\mu-r}{\sigma^{2}(1-\delta(1-k))}$, and the optimal fractional consumption $\frac{c^{*}(t)}{X^{*}(t)-\frac{d}{r}}$  takes a constant ratio as $\frac{1}{K_{1}}$.
	\end{remark}
	
	 Hereafter, we impose a stricter liquidity constraint on the wealth process, $X(t)\ge R_{post}>\frac{d}{r}$. The following proposition provides the expectation form of the liquidity constraint, which will be accessible to deduce the duality problem subsequently.
	\begin{proposition}
	The infinite horizon liquidity constraint of the post-retirement problem is
	\begin{equation}
	    \mathbb{E}\left[\left.\int_{t}^{\infty}\frac{H(s)}{H(t)}(c(s)+d)d s\right|\mathcal{F}_{t}\right]\ge R_{post}.\label{Equation 2}
	\end{equation}
	\end{proposition}
	\begin{proof}
	See \cite[Proposition 3.2]{2021arXiv210615426D}.
	\end{proof}
	
	 Referring to \cite{karatzas2000utility,he1993labor}, we introduce a real number $\lambda_{\scriptscriptstyle PR}>0$, the Lagrange multiplier, and a non-increasing process $D_{\scriptscriptstyle PR}(t)\ge 0$, then rewrite the post-retirement gain function as
	\begin{equation*}
	\begin{split}
	    J_{\scriptscriptstyle PR}(x;c,\pi)
	    \leq\mathbb{E} & \left[\int_{0}^{\infty}e^{-\gamma t}\left(\tilde{u}_{\scriptscriptstyle PR}(\lambda_{\scriptscriptstyle PR} e^{\gamma t}H(t)D_{\scriptscriptstyle PR}(t))-d\lambda_{\scriptscriptstyle PR} e^{\gamma t}D_{\scriptscriptstyle PR}(t)H(t)\right)d t\right]\\
	    &+\lambda_{\scriptscriptstyle PR} x+\lambda_{\scriptscriptstyle PR}\mathbb{E}\left[\int_{0}^{\infty}R_{post}H(t)d D_{\scriptscriptstyle PR}(t)\right].
	\end{split}
	\end{equation*}
	The derivation of this inequality involves the budget constraint (\ref{Equation 1}) and the liquidity constraint (\ref{Equation 2}). In line with \cite[Section 4]{he1993labor}, the post-retirement individual's dual shadow price problem, labelled $(S_{\scriptscriptstyle PR})$, can be defined as
	\begin{equation*}\tag{$S_{\scriptscriptstyle PR}$}
	\begin{split}
	    \tilde{V}_{\scriptscriptstyle PR}(\lambda_{\scriptscriptstyle PR})\triangleq\inf_{D_{\scriptscriptstyle PR}(t)\in\mathcal{D}} & \mathbb{E}\left[\int_{0}^{\infty}e^{-\gamma t}\left(\tilde{u}_{\scriptscriptstyle PR}(\lambda_{\scriptscriptstyle PR} e^{\gamma t}H(t)D_{\scriptscriptstyle PR}(t))-d\lambda_{\scriptscriptstyle PR} e^{\gamma t}D_{\scriptscriptstyle PR}(t)H(t)\right)d t\right]\\
	    & +\lambda_{\scriptscriptstyle PR}\mathbb{E}\left[\int_{0}^{\infty}R_{post}H(t)d D_{\scriptscriptstyle PR}(t)\right], 
	\end{split}
	\end{equation*} 
	where $\mathcal{D}$ is the set of non-negative, non-increasing and progressively measurable processes. Then the duality between Problem $(P_{\scriptscriptstyle PR})$ and Problem $(S_{\scriptscriptstyle PR})$ is put forward.
	
	\begin{theorem}
	\textbf{(Duality Theorem)}\label{Theorem 1}
	Suppose $D_{\scriptscriptstyle PR}^{*}(t)$ is the optimal solution to the dual shadow price problem $(S_{\scriptscriptstyle PR})$, then $c^{*}(t)=I_{\scriptscriptstyle PR}(\lambda_{\scriptscriptstyle PR}^{*} e^{\gamma t}D_{\scriptscriptstyle PR}^{*}(t)H(t))$ is the optimal consumption solution to the problem $(P_{\scriptscriptstyle PR})$. And we have the relation
	$V_{\scriptscriptstyle PR}(x)\!=\!\inf\limits_{\lambda_{\scriptscriptstyle PR}>0}[\tilde{V}_{\scriptscriptstyle PR}(\lambda_{\scriptscriptstyle PR})\!+\!\lambda_{\scriptscriptstyle PR} x]$, with $\lambda_{\scriptscriptstyle PR}^{*}$ attaining the infimum.
	\end{theorem}
	\begin{proof}
	See Appendix \ref{Appendix 1}.
	\end{proof}
	
	 The Duality Theorem enables us to transform the solution of Problem $(P_{\scriptscriptstyle PR})$ to its duality, $(S_{\scriptscriptstyle PR})$. Besides, adopting the technique from \cite{davis1990portfolio}, the subsequent assumption should be imposed for solving the problem explicitly.
	
	\begin{assumption}
	The non-increasing process $D_{\scriptscriptstyle PR}(t)$ is absolutely continuous with respect to t. Hence, there exists a process $\psi_{\scriptscriptstyle PR}(t)$ such that
	$d D_{\scriptscriptstyle PR}(t)=-\psi_{\scriptscriptstyle PR}(t)D_{\scriptscriptstyle PR}(t)d t$.
	\end{assumption}
	
	 Introducing $Z_{\scriptscriptstyle PR}(t)\! \triangleq\! \lambda_{\scriptscriptstyle PR} e^{\gamma t}\!D_{\scriptscriptstyle PR}(t)H(t)$, the value function of Problem $(S_{\scriptscriptstyle PR})$ is converted into
	\begin{equation*}
		\tilde{V}_{\scriptscriptstyle PR}(\lambda_{\scriptscriptstyle PR})=\inf_{\psi_{\scriptscriptstyle PR}(t)\ge 0} \mathbb{E}\left[\int_{0}^{\infty}e^{-\gamma t}(\tilde{u}_{\scriptscriptstyle PR}(Z_{\scriptscriptstyle PR}(t))-d Z_{\scriptscriptstyle PR}(t)-\psi_{\scriptscriptstyle PR}(t)Z_{\scriptscriptstyle PR}(t)R_{post})d t\right].
	\end{equation*}
	Then we define
	\begin{equation*}
	    \phi_{\scriptscriptstyle PR}(t,z)\!\triangleq\!\inf\limits_{\psi_{\scriptscriptstyle PR}(t)\ge 0}\! \mathbb{E}\left[\left.\!\int_{t}^{\infty} \!e^{\!-\!\gamma s} \left(\tilde{u}_{\scriptscriptstyle PR}(Z_{\scriptscriptstyle PR}(s))\!-\!d Z_{\scriptscriptstyle PR}(s)\!-\!\psi_{\scriptscriptstyle PR}(s)Z_{\scriptscriptstyle PR}(s)R_{post}\right)d s\right|Z_{\scriptscriptstyle PR}(t)\!=\!z\right],
	\end{equation*}
	and observe that $\tilde{V}_{\scriptscriptstyle PR}(\lambda_{\scriptscriptstyle PR})=\phi_{\scriptscriptstyle PR}(0,\lambda_{\scriptscriptstyle PR})$. The associated Bellman equation to $\phi_{\scriptscriptstyle PR}(t,z)$ follows
	\begin{equation*}
		\min_{\psi_{\scriptscriptstyle PR}\ge 0}\left\{\tilde{\mathcal{L}}\phi_{\scriptscriptstyle PR}(t,z)+e^{-\gamma t}(\tilde{u}_{\scriptscriptstyle PR}(z)-dz)-\psi_{\scriptscriptstyle PR} z\left[\frac{\partial \phi_{\scriptscriptstyle PR}}{\partial z}(t,z)+e^{-\gamma t} R_{post}\right] \right\}=\gamma\phi_{\scriptscriptstyle PR}(t,z),
	\end{equation*}
	with the operator
	$\tilde{\mathcal{L}}=(\gamma-r)z\frac{\partial}{\partial z}+\frac{1}{2}\theta^{2}z^{2}\frac{\partial^{2}}{\partial z^{2}}$.
	From the characterization of optimum $\psi_{\scriptscriptstyle PR}^{*}$:
	\begin{equation*}
	    \frac{\partial \phi_{\scriptscriptstyle PR}}{\partial z}(t,z)+e^{-\gamma t}R_{post}=0 \Rightarrow \psi_{\scriptscriptstyle PR}^{*}\ge 0; \qquad \frac{\partial \phi_{\scriptscriptstyle PR}}{\partial z}(t,z)+e^{-\gamma t}R_{post}\leq 0 \Rightarrow \psi_{\scriptscriptstyle PR}^{*}= 0,
	\end{equation*}
	the Bellman equation is equivalent to
	\begin{equation*}
		\min \left\{\tilde{\mathcal{L}}\phi_{\scriptscriptstyle PR}(t,z)-\gamma\phi_{\scriptscriptstyle PR}(t,z)+e^{-\gamma t}(\tilde{u}_{\scriptscriptstyle PR}(z)-dz),-\left[\frac{\partial \phi_{\scriptscriptstyle PR}}{\partial z}(t,z)+e^{-\gamma t}R_{post}\right]\right\}=0,
	\end{equation*}
	which results in the consequent modified variational inequalities:  Find a free boundary $\hat{z}_{\scriptscriptstyle PR}>0$, which makes $R_{post}$-wealth level, and a function $\phi_{\scriptscriptstyle PR}(\cdot,\cdot)\in C^{2}((0,\infty)\times\mathbb{R}^{+})$ satisfying
	\begin{equation}\label{Equation 3}
	\begin{cases}
	(V1) \quad\frac{\partial \phi_{\scriptscriptstyle PR}}{\partial z}(t,z)+e^{-\gamma t}R_{post}=0, & z\ge\hat{z}_{\scriptscriptstyle PR},\\
	(V2) \quad \frac{\partial \phi_{\scriptscriptstyle PR}}{\partial z}(t,z)+e^{-\gamma t}R_{post}\leq 0, & 0<z<\hat{z}_{\scriptscriptstyle PR},\\
	(V3)\quad \tilde{\mathcal{L}}\phi_{\scriptscriptstyle PR}(t,z)-\gamma\phi_{\scriptscriptstyle PR}(t,z)+e^{-\gamma t}(\tilde{u}_{\scriptscriptstyle PR}(z)-dz)=0, & 0<z<\hat{z}_{\scriptscriptstyle PR},\\
	(V4)\quad \tilde{\mathcal{L}}\phi_{\scriptscriptstyle PR}(t,z)-\gamma\phi_{\scriptscriptstyle PR}(t,z)+e^{-\gamma t}(\tilde{u}_{\scriptscriptstyle PR}(z)-dz)\ge 0, & z\ge\hat{z}_{\scriptscriptstyle PR},
	\end{cases}
	\end{equation}
	for any $t\ge 0$, with the smooth fit conditions
	$\frac{\partial \phi_{\scriptscriptstyle PR}}{\partial z}(t,\hat{z}_{\scriptscriptstyle PR})=-R_{post} e^{-\gamma t}$ and $\frac{\partial^{2}\phi_{\scriptscriptstyle PR}}{\partial z^{2}}(t,\hat{z}_{\scriptscriptstyle PR})=0$.
	
	\begin{proposition}\label{Proposition 1}
	Under the assumption $\phi_{\scriptscriptstyle PR}(t,z)=e^{-\gamma t}v_{\scriptscriptstyle PR}(z)$, the variational inequalities (\ref{Equation 3}) takes the solution 
	\begin{equation*}
	    v_{\scriptscriptstyle PR}(z)=\begin{cases}
	    B_{2,\scriptscriptstyle PR}\hat{z}_{\scriptscriptstyle PR}^{n_{2}}+\frac{1-\delta(1-k)}{\delta(1-k)}K_{1}\bar{L}^{\frac{(1-k)(1-\delta)}{1-\delta(1-k)}}\hat{z}_{\scriptscriptstyle PR}^{\frac{\delta(1-k)}{\delta(1-k)-1}}-\frac{d}{r}\hat{z}_{\scriptscriptstyle PR}-R_{post}(z-\hat{z}_{\scriptscriptstyle PR}), & z\ge\hat{z}_{\scriptscriptstyle PR},\\
	    B_{2,\scriptscriptstyle PR}z^{n_{2}}+\frac{1-\delta(1-k)}{\delta(1-k)}K_{1}\bar{L}^{\frac{(1-k)(1-\delta)}{1-\delta(1-k)}}z^{\frac{\delta(1-k)}{\delta(1-k)-1}}-\frac{d}{r}z, & 0<z<\hat{z}_{\scriptscriptstyle PR},
	\end{cases}
	\end{equation*}
	with
	\begin{equation*}
		n_{2}=-\frac{\gamma-r-\frac{\theta^{2}}{2}}{\theta^{2}}+\sqrt{\left(\frac{\gamma-r-\frac{\theta^{2}}{2}}{\theta^{2}}\right)^{2}+\frac{2\gamma}{\theta^{2}}},
	\end{equation*}
$$\hat{z}_{\scriptscriptstyle PR}=\bar{L}^{(1-k)(1-\delta)}\left[\frac{(1-n_{2})(1-\delta(1-k))}{n_{2}(\delta(1-k)-1)-\delta(1-k)}\frac{\left(R_{post}-\frac{d}{r}\right)}{K_{1}}\right]^{\delta(1-k)-1}>0,$$ 
and $$B_{2,\scriptscriptstyle PR}=\!=\!\frac{K_{1}^{(\delta(1\!-\!k)\!-\!1)(n_{2}\!-\!1)}\!\bar{L}^{(1\!-\!k)(1\!-\!\delta)(1\!-\!n_{2})}}{n_{2}(n_{2}-1)(\delta(1\!-\!k)\!-\!1)}\!\!\left[\!\frac{(1\!-\!n_{2})(1\!-\!\delta(1\!-\!k))}{n_{2}(\delta(1\!-\!k)\!-\!1)\!-\!\delta(1\!-\!k)}\!\!\left(\!R_{\scriptscriptstyle post}\!-\!\frac{d}{r}\!\right)\!\right]^{\delta(1\!-\!k)\!-\!n_{2}(\delta(1\!-\!k)\!-\!1)}\!\!\!<\!0.$$ Furthermore, for a given initial wealth $x\ge R_{\scriptscriptstyle post}$, the value function of the post-retirement problem is
	\begin{equation}
		V_{\scriptscriptstyle PR}(x)\!=\!
		B_{2,\scriptscriptstyle PR}(\lambda_{\scriptscriptstyle PR}^{*})^{n_{2}}\!+\!\frac{1\!-\!\delta(1\!-\!k)}{\delta(1\!-\!k)}K_{1}\bar{L}^{\frac{(1\!-\!k)(1\!-\!\delta)}{1\!-\!\delta(1\!-\!k)}}(\lambda_{\scriptscriptstyle PR}^{*})^{\frac{\delta(1\!-\!k)}{\delta(1\!-\!k)\!-\!1}}\!-\!\frac{d}{r}\lambda_{\scriptscriptstyle PR}^{*}\!+\!\lambda_{\scriptscriptstyle PR}^{*}x,\label{Equation 4}
	\end{equation}
	with $-n_{2}B_{2,\scriptscriptstyle PR}(\lambda_{\scriptscriptstyle PR}^{*})^{n_{2}-1}+K_{1}\bar{L}^{\frac{(1-k)(1-\delta)}{1-\delta(1-k)}}(\lambda_{\scriptscriptstyle PR}^{*})^{\frac{1}{\delta(1-k)-1}}+\frac{d}{r}=x$. Taking $Z_{\scriptscriptstyle PR}^{*}(t) \triangleq \lambda_{\scriptscriptstyle PR}^{*} e^{\gamma t}H(t)$, the optimal wealth process follows
	\begin{equation*}
		X^{*}(t)=-n_{2}B_{2,\scriptscriptstyle PR}(Z_{\scriptscriptstyle PR}^{*}(t))^{n_{2}-1}+K_{1}\bar{L}^{\frac{(1-k)(1-\delta)}{1-\delta(1-k)}}(Z_{\scriptscriptstyle PR}^{*}(t))^{\frac{1}{\delta(1-k)-1}}+\frac{d}{r},\quad 0<Z_{\scriptscriptstyle PR}^{*}(t)\leq\hat{z}_{\scriptscriptstyle PR},
	\end{equation*}
	and the corresponding optimal consumption and portfolio strategies are
	\begin{equation*}
		c^{*}(t)=I_{\scriptscriptstyle PR}(\lambda_{\scriptscriptstyle PR}^{*} e^{\gamma t}H(t))=(Z_{\scriptscriptstyle PR}^{*}(t))^{\frac{1}{\delta(1-k)-1}}\bar{L}^{\frac{(1-k)(1-\delta)}{1-\delta(1-k)}},
	\end{equation*}
	\begin{equation*}
		\pi^{*}(t)=\frac{\theta}{\sigma}\left[n_{2}(n_{2}-1)B_{2,\scriptscriptstyle PR}(Z_{\scriptscriptstyle PR}^{*}(t))^{n_{2}-1}-K_{1}\frac{1}{\delta(1-k)-1}\bar{L}^{\frac{(1-k)(1-\delta)}{1-\delta(1-k)}}(Z_{\scriptscriptstyle PR}^{*}(t))^{\frac{1}{\delta(1-k)-1}}\right].
	\end{equation*}
	\end{proposition}
	\begin{proof}
	    See Appendix \ref{Appendix 2}.
	\end{proof}
	
	Then based on the dynamic programming principle, we can only consider a subset of the admissible control set of Problem $(P)$, that is $\mathcal{A}_{1}(x)\subset\mathcal{A}(x)$, in which any policy achieves the maximum of the post-retirement problem's gain function. Hence, for any $(\tau, \{c(t),\pi(t),l(t) \})\in\mathcal{A}_{1}(x)$, we have
	$\mathbb{E}\left[\int_{\tau}^{\infty}e^{-\gamma t}u(c(t),\bar{L})d t\right]=\mathbb{E}\left[e^{-\gamma \tau}V_{\scriptscriptstyle PR}(X^{x,c,\pi,l}(\tau))\mathbb{I}_{\{\tau<\infty\}}\right]$.
	Afterwards, the whole optimization problem can be rewritten as
	\begin{equation*}
		V(x)= \sup_{\left(\tau, \{c(t),\pi(t),l(t) \}\right)\in\mathcal{A}_{1}(x)} \mathbb{E}\left[\int_{0}^{\tau}e^{-\gamma t}u(c(t),l(t))d t+e^{-\gamma\tau}U\left(X^{x,c,\pi,l}(\tau)\right)\right],
	\end{equation*} 
	denoting $U\left(X^{x,c,\pi,l}(\tau)\right)\!\triangleq\! \!\sup\limits_{ \{c(t),\pi(t),l(t) \}\in\mathcal{A}_{1}(x)}\!\! \mathbb{E}\!\left[\!\left.\int_{\tau}^{\infty}\!\!e^{\!-\!\gamma (s\!-\!\tau)}u(c(s),\bar{L})d s\right|\mathcal{F}_{\tau}\right]\!=\!V_{\scriptscriptstyle PR}\left(X^{x,c,\pi,l}(\tau)\right)$.
	Finally, we summarize the two different forms of $U(x)$ and introduce its Legendre-Fenchel transform under the definition $\tilde{U}(z)\!\triangleq\!\!\sup\limits_{x\ge R_{post}}\!\!\left[U(x)\!-\!x z\right]$, $0\!<\!z\!<\!\infty$ following \cite[Chapter 3, Definition 4.2]{karatzas1998methods}.

	\begin{lemma}\label{Lemma 3b}
		The post-retirement value function $U(x)$, for $x\ge R_{post}$, is given in two separate cases, the partition being based on the value of threshold in the liquidity constraint, i.e., $R_{post}$.
		\begin{equation*}
			U\left(x\right)=\begin{cases}
				\left(x-\frac{d}{r}\right)^{\delta(1-k)}K_{1}^{1-\delta(1-k)}\bar{L}^{(1-k)(1-\delta)}\frac{1}{\delta(1-k)}, &  \mbox{if}\quad R_{post}=\frac{d}{r},\\
				B_{2,\scriptscriptstyle PR}(\lambda_{\scriptscriptstyle PR}^{*})^{n_{2}}+\frac{1-\delta(1-k)}{\delta(1-k)}K_{1}\bar{L}^{\frac{(1-k)(1-\delta)}{1-\delta(1-k)}}(\lambda_{\scriptscriptstyle PR}^{*})^{\frac{\delta(1-k)}{\delta(1-k)-1}}-\frac{d}{r}\lambda_{\scriptscriptstyle PR}^{*}+\lambda_{\scriptscriptstyle PR}^{*}x,& \mbox{if}\quad R_{post}>\frac{d}{r}.
			\end{cases}
		\end{equation*}
		Furthermore, the Legendre-Fenchel transform of $U(x)$ is:
		\begin{itemize}
			\item $\tilde{U}(z)=\frac{1-\delta(1-k)}{\delta(1-k)}z^{\frac{\delta(1-k)}{\delta(1-k)-1}}K_{1}\bar{L}^{\frac{(1-k)(1-\delta)}{1-\delta(1-k)}}-\frac{d}{r}z$, $z>0$, if $R_{post}\!=\!\frac{d}{r}$;
			\item 
				$\tilde{U}(z)\!=\!\begin{cases}
					\!B_{2,\scriptscriptstyle PR}\hat{z}_{\scriptscriptstyle PR}^{n_{2}}\!+\!\frac{1\!-\!\delta(1\!-\!k)}{\delta(1\!-\!k)}K_{1}\!\bar{L}^{\frac{(1\!-\!k)(1\!-\!\delta)}{1\!-\!\delta(1\!-\!k)}}\hat{z}_{\scriptscriptstyle PR}^{\frac{\delta(1\!-\!k)}{\delta(1\!-\!k)\!-\!1}}\!\!\!\!-\!\frac{d}{r}\hat{z}_{\scriptscriptstyle PR}\!\!\!-\!R_{post}(z\!-\!\hat{z}_{\scriptscriptstyle PR}), & z\!\ge\!\hat{z}_{\scriptscriptstyle PR},\\
					\!B_{2,\scriptscriptstyle PR}z^{n_{2}}+\frac{1-\delta(1-k)}{\delta(1-k)}K_{1}\bar{L}^{\frac{(1-k)(1-\delta)}{1-\delta(1-k)}}z^{\frac{\delta(1-k)}{\delta(1-k)-1}}-\frac{d}{r}z, & 0\!\!<\!z\!\!<\!\hat{z}_{\scriptscriptstyle PR},
				\end{cases}$ if $R_{post}\!>\!\frac{d}{r}$.
		\end{itemize}
	\end{lemma}	
	\begin{proof}
	    See Appendix \ref{Appendix 3}.
	\end{proof}

	 \subsection{Proof of Theorem \ref{Theorem 1}}\label{Appendix 1}
	 \noindent 
	 We first provide a lemma for proving Theorem \ref{Theorem 1}.
	 
	 \begin{lemma}\label{Lemma Appendix 1}
	 For any given initial wealth $x>R_{\scriptscriptstyle post}$, and any given progressively measurable consumption process $c(t)\ge 0$ satisfying $\sup\limits_{\tau\in\mathcal{T}}\mathbb{E}\left[\int_{0}^{\tau}H(t)(c(t)+d)d t\right]\leq x-R_{\scriptscriptstyle post}$, with $\mathcal{T}$ standing for the set of $\mathcal{F}$-stopping times, there exists a portfolio process $\pi(t)$ making $X^{x,c,\pi}(t)\ge R_{\scriptscriptstyle post}$, $\forall t\ge 0$, holds almost surely.
	\end{lemma}
	
	\begin{proof}
		Adopting the technique of \cite[Appendix, Lemma 1]{he1993labor}, we introduce $K(t)\triangleq\int_{0}^{t}(c(s)+d)H(s)d s$ and show that $\{K(\tau)\}_{\tau\in\mathcal{T}}$ is uniformly integrable from the fact $\mathbb{E}[K(t)]<\infty$. Then, Dellacherie \& Meyer (1982), Appendix I,\footnote{C. Dellacherie and P. Meyer, Probabilities and potential b, theory of martingales, North–Holland Mathematics Studies, 1982.} indicates that there exists a Snell envelope of $K(t)$ denoted as $\bar{K}(t)$, which is a super-martingale under the $\mathbb{P}$ measure and satisfies $\bar{K}(0)=\sup\limits_{\tau\in\mathcal{T}}\mathbb{E}[K(\tau)]$, $\bar{K}(\infty)=K(\infty)$. The Doob-Meyer Decomposition Theorem of Karatzas \& Shreve (1998), Section 1.4, Theorem 4.10,\footnote{I. Karatzas and S. E. Shreve, Brownian Motion and Stochastic Calculus. Second edition. Springer-Verlag, 1998.} enables us to represent the super-martingale $\bar{K}(t)$ as $\bar{K}(t)=\bar{K}(0)+\bar{M}(t)-\bar{A}(t)$,
		with a uniformly integrable martingale under the $\mathbb{P}$ measure $\bar{M}(t)$ satisfying $\bar{M}(0)=0$ and a strictly increasing process $\bar{A}(t)$ satisfying $\bar{A}(0)=0$. Moreover, the Martingale Representation Theorem from Bjork (2009), Chapter 11, Theorem 11.2,\footnote{T. Bjork, Arbitrage theory in continuous time. Oxford university press, 2009.} makes $\bar{M}(t)$ take expression $\bar{M}(t)=\int_{0}^{t}\bar{\rho}(s)dB(s)$,
		where $\bar{\rho}(t)$ is an $\mathbb{F}$-adapted process satisfying $\int_{0}^{\infty}\bar{\rho}^{2}(s)d s<\infty$ a.s..
		
		\noindent Defining a new process $\bar{X}(t)\triangleq\frac{1}{H(t)}\left[x-\bar{K}(0)+\bar{K}(t)-K(t)+\bar{A}(t)\right]-R_{\scriptscriptstyle post}$, it can be observed that $\bar{X}(t)$ is a non-negative process with the initial wealth $\bar{X}(0)=x-R_{\scriptscriptstyle post}$,
		because of $$\bar{K}(0)=\sup\limits_{\tau\in\mathcal{T}}\mathbb{E}[K(\tau)]=	\sup\limits_{\tau\in\mathcal{T}}\mathbb{E}\left[\int_{0}^{\tau}H(t)(c(t)+d)d t\right]\leq x-R_{\scriptscriptstyle post}.$$ Then $\bar{X}(t)$ can be expressed with $\bar{M}(t)$ as
		\begin{equation*}
			\bar{X}(t)=\frac{1}{H(t)}\left[x+\bar{M}(t)-K(t)\right]-R_{\scriptscriptstyle post}=\frac{1}{H(t)}\left[x+\int_{0}^{t}\bar{\rho}(s)dB(s)-\int_{0}^{t}(c(s)+d)H(s)d s\right]-R_{\scriptscriptstyle post}.
		\end{equation*}
		Applying the It\^o's formula to $H(t)X^{x,c,\pi}(t)$, we can get
		\begin{equation*}
			d(H(t)X^{x,c,\pi}(t))=-H(t)X^{x,c,\pi}(t)\theta dB(t)-(c(t)+d)H(t)d t+\sigma\pi(t)H(t)dB(t).
		\end{equation*}
		Considering the portfolio strategy $\pi(t)=\frac{\bar{\rho}(t)}{\sigma H(t)}+\frac{\theta X^{x,c,\pi}(t)}{\sigma}$, the wealth process takes
		\begin{equation*}
		    X^{x,c,\pi}(t)=\frac{1}{H(t)}\left[x+\int_{0}^{t}\bar{\rho}(s)dB(s)-\int_{0}^{t}(c(s)+d)H(s)d s\right],
		\end{equation*}
		which indicates that $\bar{X}(t)=X^{x,c,\pi}(t)-R_{\scriptscriptstyle post}$, a.s..
		The non-negativity of $\bar{X}(t)$ makes clear that $X^{x,c,\pi}(t)\ge R_{\scriptscriptstyle post}$, a.s., $\forall t\ge 0.$
	\end{proof}
	
	 \noindent Now we turn back to the proof of Theorem \ref{Theorem 1}. Following \cite[Section 4, Theorem 1]{he1993labor}, the proof mainly contains two aspects: the first part is to show the admissibility of $c^{*}(t)$, and the second part is to claim that $c^{*}(t)$ is the optimal consumption strategy to Problem $(P_{\scriptscriptstyle PR})$.
	 
	 \noindent (1) We first prove that $c^{*}(t)=I_{\scriptscriptstyle PR}^{*}(\lambda_{\scriptscriptstyle PR}^{*} e^{\gamma t}D_{\scriptscriptstyle PR}^{*}(t)H(t))$ is an admissible consumption policy. Taking any stopping time $\tau$ from $\mathcal{T}$ and a positive constant $\epsilon$, we can introduce $D_{\scriptscriptstyle PR}^{\epsilon}(t)\triangleq D_{\scriptscriptstyle PR}^{*}(t)+\epsilon\mathbb{I}_{[0,\tau)}(t)$, which evidently satisfies $D_{\scriptscriptstyle PR}^{\epsilon}(t)\in\mathcal{D}$. Then defining a function
		\begin{equation*}
		\begin{split}
		\mathfrak{L}(D(t))\triangleq\mathbb{E} & \left[\int_{0}^{\infty}e^{-\gamma t}\big(\tilde{u}_{\scriptscriptstyle PR}(\lambda _{\scriptscriptstyle PR}^{*}e^{\gamma t}D(t)H(t))-d\lambda_{\scriptscriptstyle PR}^{*} e^{\gamma t} D(t) H(t)\big)d t\right]\\
		& +\lambda_{\scriptscriptstyle PR}^{*}\mathbb{E}\left[\int_{0}^{\infty}R_{\scriptscriptstyle post}H(t)d D(t)\right]+\lambda_{\scriptscriptstyle PR}^{*}(x-R_{\scriptscriptstyle post})D(0),
		\end{split}
		\end{equation*}
		an inequality, $\mathfrak{L}(D_{\scriptscriptstyle PR}^{*}(t))\leq\mathfrak{L}(D_{\scriptscriptstyle PR}^{\epsilon}(t))$, is obtained from the facts that $D_{\scriptscriptstyle PR}^{*}(t)$ is the optimal solution of Problem $(S_{\scriptscriptstyle PR})$ and $x\ge R_{\scriptscriptstyle post}$.
		This inequality gives us
		\begin{equation*}
			\limsup_{\epsilon\downarrow 0}\mathbb{E}\!\left[\!\int_{0}^{\tau}\!\!\!\!\!\!\big(e^{\!-\!\gamma t}\!\frac{\tilde{u}_{\scriptscriptstyle PR}(\lambda_{\scriptscriptstyle PR}^{*} e^{\gamma t}\!D_{\scriptscriptstyle PR}^{\epsilon }(t)H(t))\!\!-\!\tilde{u}_{\scriptscriptstyle PR}(\lambda_{\scriptscriptstyle PR}^{*} e^{\gamma t}\!D_{\scriptscriptstyle PR}^{*}(t)H(t))}{\epsilon}\!-\!d\lambda_{\scriptscriptstyle PR}^{*} H(t)\big)d t\!\right]\!+\!\lambda_{\scriptscriptstyle PR}^{*}(x\!-\!R_{\scriptscriptstyle post})\!\!\ge\!0,
		\end{equation*} 
		considering $d D^{\epsilon}(t)=d D^{*}(t)$, $\forall t\in(0,\tau)$. The decreasing property of $\tilde{u}_{\scriptscriptstyle PR}(\cdot)$ and the Fatou's lemma endows us with
		\begin{equation*}
		\begin{split}
		 \mathbb{E}\bigg[\!\int_{0}^{\tau} & e^{-\gamma t}\tilde{u}_{\scriptscriptstyle PR}^{\prime}(\lambda_{\scriptscriptstyle PR}^{*} e^{\gamma t}D_{\scriptscriptstyle PR}^{*}(t)H(t))\lambda_{\scriptscriptstyle PR}^{*} e^{\gamma t}H(t)d t\bigg]\ge \\
		 & \limsup_{\epsilon\downarrow 0}\mathbb{E}\left[\int_{0}^{\tau}e^{-\gamma t}\frac{\tilde{u}_{\scriptscriptstyle PR}(\lambda_{\scriptscriptstyle PR}^{*} e^{\gamma t}\!D_{\scriptscriptstyle PR}^{\epsilon }(t)H(t))-\tilde{u}_{\scriptscriptstyle PR}(\lambda_{\scriptscriptstyle PR}^{*} e^{\gamma t}D_{\scriptscriptstyle PR}^{*}(t)H(t))}{\epsilon}d t\right].
		\end{split}
		\end{equation*}
		Then $\tilde{u}_{\scriptscriptstyle PR}^{\prime}(\cdot)=-I_{\scriptscriptstyle PR}(\cdot)$ indicates that $\mathbb{E}\left[\int_{0}^{\tau}H(t)(c^{*}(t)+d)d t\right]\leq x-R_{\scriptscriptstyle post}$.
		Since $\tau$ can be any stopping time in the set $\mathcal{T}$, Lemma \ref{Lemma Appendix 1} claims that there exists a portfolio strategy $\pi^{*}(t)$ making the related wealth process satisfying $X^{x,c^{*},\pi^{*}}(t)\ge R_{\scriptscriptstyle post}$, $\forall t\ge 0$.
		
		\noindent (2) We move to show the optimality of $c^{*}(t)$ to Problem $(P_{\scriptscriptstyle PR})$. The proof of Lemma \ref{Lemma Appendix 1} indicates that for an arbitrary consumption strategy $c(t)\in\mathcal{A}_{\scriptscriptstyle PR}(x)$, there exists a process $\zeta(t)$ satisfying
		\begin{equation}
			\int_{0}^{t}(c(s)+d)H(s)d s+H(t)X^{x,c,\pi}(t)=x+\int_{0}^{t}\zeta(s)dB(s).\label{Equation Appendix 1}
		\end{equation}
		The property $X^{x,c,\pi}(t)\ge R_{\scriptscriptstyle post}$ a.s. gives us the subsequent inequality with any process $D(t)\in\mathcal{D}$,
		\begin{equation*}
			\int_{0}^{T}\int_{0}^{t}(c(s)+d)H(s)d s d D(t)+\int_{0}^{T}R_{\scriptscriptstyle post}H(t)d D(t)\ge\int_{0}^{T}\left[x+\int_{0}^{t}\zeta(s)dB(s)\right]d D(t),
		\end{equation*}
		where $T$ is any time meeting $T\!\ge\! t$. Since $D(t)$ is bounded variational, integrating by parts gives us
		\begin{equation*}
			\begin{split}
				\int_{0}^{T}D(s) & (c(s)+d)H(s)d s -\int_{0}^{T}D(s)\zeta(s)dB(s)\leq \\ &D(0)x+D(T)\left[\int_{0}^{T}(c(s)+d)H(s)d s\!-\!x\!-\!\int_{0}^{T}\zeta(s)dB(s)\right]+\int_{0}^{T}R_{\scriptscriptstyle post}H(s)d D(s).
			\end{split}
		\end{equation*}
		Then we can take the expectation under the $\mathbb{P}$ measure on both sides and replace Equation (\ref{Equation Appendix 1}) to get $\mathbb{E}\left[\int_{0}^{T}D(s)(c(s)+d)H(s)d s\right]\leq D(0)x+\mathbb{E}\left[\int_{0}^{T}R_{\scriptscriptstyle post}H(s)d D(s)\right]$. Then the Lebesgue's Monotone Convergence Theorem indicates that
		\begin{equation*}
		    \mathbb{E}\left[\int_{0}^{\infty}D(s)(c(s)+d)H(s)d s\right]\leq D(0)x+\mathbb{E}\left[\int_{0}^{\infty}R_{\scriptscriptstyle post}H(s)d D(s)\right],
		\end{equation*}
		which holds for any admissible consumption policy $c(t)$ and any non-negative, non-increasing process $D(t)$. Furthermore, it will be proved that the above inequality  becomes equalized with the given $c^{*}(t)$ and $D_{\scriptscriptstyle PR}^{*}(t)$. Introducing $\bar{D}_{\scriptscriptstyle PR}^{\epsilon}(t)\triangleq D_{\scriptscriptstyle PR}^{*}(t)(1+\epsilon)\in\mathcal{D}$ with a small enough constant $\epsilon$ and defining a new function as
		\begin{equation*}
		\begin{split}
		\tilde{\mathfrak{L}}(D(t))\triangleq \mathbb{E} & \left[\int_{0}^{\infty}e^{-\gamma t}\big(\tilde{u}_{\scriptscriptstyle PR}(\lambda_{\scriptscriptstyle PR}^{*} e^{\gamma t}D(t)H(t))-d\lambda_{\scriptscriptstyle PR}^{*} e^{\gamma t}D(t)H(t)\big)d t\right]\\
		& +\lambda_{\scriptscriptstyle PR}^{*}\mathbb{E}\left[\int_{0}^{\infty}R_{\scriptscriptstyle post}H(t)d D(t)\right]+\lambda_{\scriptscriptstyle PR}^{*} x D(0),
		\end{split}
		\end{equation*}
		we get $\tilde{\mathfrak{L}}(\bar{D}_{\scriptscriptstyle PR}^{\epsilon}(t))\ge\tilde{\mathfrak{L}}(D_{\scriptscriptstyle PR}^{*}(t))$. Following the same argument with the first part, we apply the Fatou's lemma to obtain separately
		\begin{equation*}
			\begin{split}
				& \mathbb{E}\left[\int_{0}^{\infty}D_{\scriptscriptstyle PR}^{*}(t)H(t)(c^{*}(t)+d)d t\right]\leq x D_{\scriptscriptstyle PR}^{*}(0)+\mathbb{E}\left[\int_{0}^{\infty}R_{\scriptscriptstyle post}H(t)d D_{\scriptscriptstyle PR}^{*}(t)\right],\\
				&\mathbb{E}\left[\int_{0}^{\infty}D_{\scriptscriptstyle PR}^{*}(t)H(t)(c^{*}(t)+d)d t\right]\ge x D_{\scriptscriptstyle PR}^{*}(0)+\mathbb{E}\left[\int_{0}^{\infty}R_{\scriptscriptstyle post}H(t)d D_{\scriptscriptstyle PR}^{*}(t)\right],
			\end{split}
		\end{equation*}
		which claims $\mathbb{E}\left[\int_{0}^{\infty}\!D_{\scriptscriptstyle PR}^{*}(t)H(t)(c^{*}(t)\!+\!d)d t\right]\!=\!x D_{\scriptscriptstyle PR}^{*}(0)\!+\!\mathbb{E}\left[\int_{0}^{\infty}\!R_{\scriptscriptstyle post}H(t)d D_{\scriptscriptstyle PR}^{*}(t)\right]$.
		Afterwards, we define a new optimization problem named $(P_{\scriptscriptstyle PR}^{\prime})$ as
		\begin{equation*}
			\max_{c(t)\ge 0} \mathbb{E}\left[\int_{0}^{\infty}e^{-\gamma t}u_{\scriptscriptstyle PR}(c(t))d t\right] \tag{$P_{\scriptscriptstyle PR}^{\prime}$}
		\end{equation*}
		\begin{equation*}
			s.t. \quad\mathbb{E} \left[\int_{0}^{\infty}D_{\scriptscriptstyle PR}^{*}(t)H(t)(c(t)+d)d t\right]\leq x D_{\scriptscriptstyle PR}^{*}(0)+\mathbb{E}\left[\int_{0}^{\infty}R_{\scriptscriptstyle post}H(t)d D_{\scriptscriptstyle PR}^{*}(t)\right].
		\end{equation*}
		The Lagrange method implies that the optimal consumption solution of the above problem, denoted as $\tilde{c}^{*}(t)$, satisfies $e^{-\gamma t}u_{\scriptscriptstyle PR}^{\prime}(\tilde{c}^{*}(t))=\tilde{\lambda}_{\scriptscriptstyle PR} D_{\scriptscriptstyle PR}^{*}(t)H(t)$, with $\tilde{\lambda}_{\scriptscriptstyle PR}>0$ as the Lagrange multiplier. The condition $\tilde{\lambda}_{\scriptscriptstyle PR}=\lambda_{\scriptscriptstyle PR}^{*}$ makes the constraint of Problem $(P_{\scriptscriptstyle PR}^{\prime})$ takes equality. And the condition $u_{\scriptscriptstyle PR}^{\prime}(\tilde{c}^{*}(t))=\lambda_{\scriptscriptstyle PR}^{*} e^{\gamma t}D_{\scriptscriptstyle PR}^{*}(t)H(t)$ implies that $\tilde{c}^{*}(t)=I(\lambda_{\scriptscriptstyle PR}^{*} e^{\gamma t}D_{\scriptscriptstyle PR}^{*}(t)H(t))=c^{*}(t)$,
		which shows that $c^{*}(t)$ is the optimal consumption policy of Problem $(P_{\scriptscriptstyle PR}^{\prime})$. Finally, since the maximum utility of Problem $(P_{\scriptscriptstyle PR})$ is upper bounded by the maximum utility of $(P_{\scriptscriptstyle PR}^{\prime})$, we can conclude that $c^{*}(t)$ is also the optimal consumption solution of the primal problem $(P_{\scriptscriptstyle PR})$.
		
		\subsection{Proof of Proposition \ref{Proposition 1}}\label{Appendix 2}
		\noindent Referring to \cite[Appendix A]{choi2008optimal}, the function $\phi_{\scriptscriptstyle PR}(t,z)$ is assumed to be time-independent, that is, $\phi_{\scriptscriptstyle PR}(t,z)=e^{-\gamma t}v_{\scriptscriptstyle PR}(z)$. Then the condition $(V3)$ of (\ref{Equation 3}) leads to a differential equation
		\begin{equation}
		-\gamma v_{\scriptscriptstyle PR}(z)+(\gamma-r)z v_{\scriptscriptstyle PR}^{\prime}(z)+\frac{1}{2}\theta^{2}z^{2}v_{\scriptscriptstyle PR}^{\prime\prime}(z)+\tilde{u}_{\scriptscriptstyle PR}(z)-dz=0,\quad 0<z<\hat{z}_{\scriptscriptstyle PR},\label{Equation Appendix 3}
		\end{equation}
		which has the solution 
		\begin{equation*}
		v_{\scriptscriptstyle PR}(z)=
		B_{1,\scriptscriptstyle PR}z^{n_{1}}+B_{2,\scriptscriptstyle PR}z^{n_{2}}+\frac{1-\delta(1-k)}{\delta(1-k)}K_{1}\bar{L}^{\frac{(1-k)(1-\delta)}{1-\delta(1-k)}}z^{\frac{\delta(1-k)}{\delta(1-k)-1}}-\frac{d}{r}z, \quad 0<z<\hat{z}_{\scriptscriptstyle PR}.
		\end{equation*}
		$n_{1}$ and $n_{2}$ are the roots of the second-order equation
		$\frac{\theta^{2}}{2}n^{2}+\left(\gamma-r-\frac{\theta^{2}}{2}\right)n-\gamma=0$,
		and satisfy 
		\begin{equation*}
		n_{1,2}=-\frac{\gamma-r-\frac{\theta^{2}}{2}}{\theta^{2}}\mp\sqrt{\left(\frac{\gamma-r-\frac{\theta^{2}}{2}}{\theta^{2}}\right)^{2}+\frac{2\gamma}{\theta^{2}}},\qquad n_{1}<0,\quad n_{2}>1.
		\end{equation*}
		Since $n_{1}<0$, the term $z^{n_{1}}$ will suffer the explosion as $z$ goes to 0. Therefore, we set the coefficient $B_{1,\scriptscriptstyle PR}=0$ by the boundedness assumption. Considering the smooth conditions at $\hat{z}_{\scriptscriptstyle PR}$, we can construct a two-equations system to determine the parameters $B_{2,\scriptscriptstyle PR}$ and $\hat{z}_{\scriptscriptstyle PR}$.
		\begin{itemize}
		\item $\mathcal{C}^{1}$ condition at $z=\hat{z}_{\scriptscriptstyle PR}$:
		$n_{2}B_{2,\scriptscriptstyle PR}\hat{z}_{\scriptscriptstyle PR}^{n_{2}-1}-K_{1}\bar{L}^{\frac{(1-k)(1-\delta)}{1-\delta(1-k)}}\hat{z}_{\scriptscriptstyle PR}^{\frac{1}{\delta(1-k)-1}}-\frac{d}{r}+R_{post}=0.$
		\item $\mathcal{C}^{2}$ condition at $z=\hat{z}_{\scriptscriptstyle PR}$:
		$n_{2}(n_{2}-1)B_{2,\scriptscriptstyle PR}\hat{z}_{\scriptscriptstyle PR}^{n_{2}-2}-K_{1}\frac{1}{\delta(1-k)-1}\bar{L}^{\frac{(1-k)(1-\delta)}{1-\delta(1-k)}}\hat{z}_{\scriptscriptstyle PR}^{\frac{2-\delta(1-k)}{\delta(1-k)-1}}=0.$
	\end{itemize}
	By multiplying the $\mathcal{C}^{2}$ condition with $\hat{z}_{\scriptscriptstyle PR}$ and then adding with the $\mathcal{C}^{1}$ condition, we have 
	\begin{equation*}
		B_{2,\scriptscriptstyle PR}\hat{z}_{\scriptscriptstyle PR}^{n_{2}-1}=K_{1}\frac{\delta(1-k)}{\delta(1-k)-1}\frac{1}{n_{2}^{2}}\bar{L}^{\frac{(1-k)(1-\delta)}{1-\delta(1-k)}}\hat{z}_{\scriptscriptstyle PR}^{\frac{1}{\delta(1-k)-1}}+\frac{1}{n_{2}^{2}}\left(\frac{d}{r}-R_{post}\right).
	\end{equation*}
	Then, substituting the above expression into the $\mathcal{C}^{1}$ condition, we get  the exact value of $\hat{z}_{\scriptscriptstyle PR}$ as
	\begin{equation*}
		\hat{z}_{\scriptscriptstyle PR}=\bar{L}^{(1-k)(1-\delta)}\left[\frac{(1-n_{2})(1-\delta(1-k))}{n_{2}(\delta(1-k)-1)-\delta(1-k)}\frac{\left(R_{post}-\frac{d}{r}\right)}{K_{1}}\right]^{\delta(1-k)-1}>0,
	\end{equation*}
	and $B_{2,\scriptscriptstyle PR}$ can also be solved by bringing $\hat{z}_{\scriptscriptstyle PR}$ into the expression $B_{2,\scriptscriptstyle PR}\hat{z}_{\scriptscriptstyle PR}^{n_{2}-1}$,
	\begin{equation*}
			B_{2,\scriptscriptstyle PR}\!=\!\frac{K_{1}^{(\delta(1\!-\!k)\!-\!1)(n_{2}\!-\!1)}\!\bar{L}^{(1\!-\!k)(1\!-\!\delta)(1\!-\!n_{2})}}{n_{2}(n_{2}-1)(\delta(1\!-\!k)\!-\!1)}\!\!\left[\!\frac{(1\!-\!n_{2})(1\!-\!\delta(1\!-\!k))}{n_{2}(\delta(1\!-\!k)\!-\!1)\!-\!\delta(1\!-\!k)}\!\!\left(\!R_{\scriptscriptstyle post}\!-\!\frac{d}{r}\!\right)\!\right]^{\delta(1\!-\!k)\!-\!n_{2}(\delta(1\!-\!k)\!-\!1)}\!\!\!<\!0.
	\end{equation*}
	Moreover, the piecewise function of $v_{\scriptscriptstyle PR}(z)$ is completely determined as
	\begin{equation*}
		v_{\scriptscriptstyle PR}(z)=\begin{cases}
			B_{2,\scriptscriptstyle PR}\hat{z}_{\scriptscriptstyle PR}^{n_{2}}+\frac{1-\delta(1-k)}{\delta(1-k)}K_{1}\bar{L}^{\frac{(1-k)(1-\delta)}{1-\delta(1-k)}}\hat{z}_{\scriptscriptstyle PR}^{\frac{\delta(1-k)}{\delta(1-k)-1}}-\frac{d}{r}\hat{z}_{\scriptscriptstyle PR}-R_{post}(z-\hat{z}_{\scriptscriptstyle PR}), & z\ge\hat{z}_{\scriptscriptstyle PR},\\
			B_{2,\scriptscriptstyle PR}z^{n_{2}}+\frac{1-\delta(1-k)}{\delta(1-k)}K_{1}\bar{L}^{\frac{(1-k)(1-\delta)}{1-\delta(1-k)}}z^{\frac{\delta(1-k)}{\delta(1-k)-1}}-\frac{d}{r}z, & 0<z<\hat{z}_{\scriptscriptstyle PR}.
		\end{cases}
	\end{equation*}
	
	\noindent Since $v_{\scriptscriptstyle PR}(z)$ is a piecewise polynomial function with smoothing merging conditions and differentiable everywhere, \cite[Section 8, Theorem 8.5]{karatzas2000utility} indicates that $V_{\scriptscriptstyle PR}(x)=\inf\limits_{\lambda_{\scriptscriptstyle PR}>0}[\tilde{V}_{\scriptscriptstyle PR}(\lambda_{\scriptscriptstyle PR})+\lambda_{\scriptscriptstyle PR} x]$ keeps true for any given initial wealth $x\ge R_{post}$. Thereafter, the closed-form of $V_{\scriptscriptstyle PR}(x)$ is
	\begin{equation*}
		V_{\scriptscriptstyle PR}(x)\!=\!
		B_{2,\scriptscriptstyle PR}(\lambda_{\scriptscriptstyle PR}^{*})^{n_{2}}\!+\!\frac{1\!-\!\delta(1\!-\!k)}{\delta(1\!-\!k)}K_{1}\bar{L}^{\frac{(1\!-\!k)(1\!-\!\delta)}{1\!-\!\delta(1\!-\!k)}}(\lambda_{\scriptscriptstyle PR}^{*})^{\frac{\delta(1\!-\!k)}{\delta(1\!-\!k)\!-\!1}}\!-\!\frac{d}{r}\lambda_{\scriptscriptstyle PR}^{*}\!+\!\lambda_{\scriptscriptstyle PR}^{*}x, \quad x\!\ge\!\hat{x}_{\scriptscriptstyle PR},
	\end{equation*}
	with $-n_{2}B_{2,\scriptscriptstyle PR}(\lambda_{\scriptscriptstyle PR}^{*})^{n_{2}-1}+K_{1}\bar{L}^{\frac{(1-k)(1-\delta)}{1-\delta(1-k)}}(\lambda_{\scriptscriptstyle PR}^{*})^{\frac{1}{\delta(1-k)-1}}+\frac{d}{r}=x$, $x\ge\hat{x}_{\scriptscriptstyle PR}$. $\hat{x}_{\scriptscriptstyle PR}$ is the critical wealth level corresponding to $\hat{z}_{\scriptscriptstyle PR}$ and follows
	\begin{equation*}
		\hat{x}_{\scriptscriptstyle PR}=\left.-\frac{\partial v_{\scriptscriptstyle PR}}{\partial z}\right|_{z=\hat{z}_{\scriptscriptstyle PR}}=-n_{2}B_{2,\scriptscriptstyle PR}\hat{z}_{\scriptscriptstyle PR}^{n_{2}-1}+K_{1}\bar{L}^{\frac{(1-k)(1-\delta)}{1-\delta(1-k)}}\hat{z}_{\scriptscriptstyle PR}^{\frac{1}{\delta(1-k)-1}}+\frac{d}{r}.
	\end{equation*}
	Moreover, the optimal wealth process takes the form
	\begin{equation*}
	    X^{*}(t)\!=\!-\!v_{\scriptscriptstyle PR}^{\prime}(Z_{\scriptscriptstyle PR}^{*}(t))\!=\!-\!n_{2}B_{2,\scriptscriptstyle PR}(Z_{\scriptscriptstyle PR}^{*}(t))^{n_{2}\!-\!1}\!+\!K_{1}\bar{L}^{\frac{(1\!-\!k)(1\!-\!\delta)}{1\!-\!\delta(1\!-\!k)}}(Z_{\scriptscriptstyle PR}^{*}(t))^{\frac{1}{\delta(1\!-\!k)\!-\!1}}\!+\!\frac{d}{r},\quad 0\!<\!Z_{\scriptscriptstyle PR}^{*}(t)\!\leq\!\hat{z}_{\scriptscriptstyle PR},
	\end{equation*}
	and the related optimal consumption-portfolio strategies are
	\begin{equation*}
		c^{*}(t)=I_{\scriptscriptstyle PR}(\lambda_{\scriptscriptstyle PR}^{*} e^{\gamma t}H(t))=(Z_{\scriptscriptstyle PR}^{*}(t))^{\frac{1}{\delta(1-k)-1}}\bar{L}^{\frac{(1-k)(1-\delta)}{1-\delta(1-k)}},
	\end{equation*}
	\begin{equation*}
		\pi^{*}(t)=\frac{\theta}{\sigma}\left[n_{2}(n_{2}-1)B_{2,\scriptscriptstyle PR}(Z_{\scriptscriptstyle PR}^{*}(t))^{n_{2}-1}-K_{1}\frac{1}{\delta(1-k)-1}\bar{L}^{\frac{(1-k)(1-\delta)}{1-\delta(1-k)}}(Z_{\scriptscriptstyle PR}^{*}(t))^{\frac{1}{\delta(1-k)-1}}\right],
	\end{equation*}
	the optimal portfolio strategy is obtained from \cite[Section 5, Theorem 3]{he1993labor}.
		
	\subsection{Proof of Lemma \ref{Lemma 3b}}\label{Appendix 3}
	
	\noindent The form of function $U(x)$ is directly summarized from Equation (\ref{Equation 5}) and (\ref{Equation 4}), hence the proof here only focuses on the derivation of the Legendre-Fenchel transform of $U(x)$, which is also divided into two cases. We first extend the supremum in the definition of Legendre-Fenchel transform $\tilde{U}(z)$ by enlarging the range of $x$ to $\mathbb{R}$, that is, $\tilde{U}(z)=\sup\limits_{x\in\mathbb{R}}[U(x)-x z]$, for $0<z<\infty$.
	Moreover, it can be proved the optimal solution $x^{*}$ attaining the supremum automatically satisfies $x^{*}\ge R_{post}$.
	\\(1) $R_{post}=\frac{d}{r}$: From the first-order condition, we have
	\begin{equation*}
	z=U^{\prime}(x^{*})=K_{1}^{1-\delta(1-k)}\bar{L}^{(1-k)(1-\delta)}\left(x^{*}-\frac{d}{r}\right)^{\delta(1-k)-1},
	\end{equation*}
	which entails that $x^{*}=\left(z K_{1}^{\delta(1-k)-1}\bar{L}^{-(1-k)(1-\delta)}\right)^{\frac{1}{\delta(1-k)-1}}+\frac{d}{r}$. Then $x^{*}>R_{post}=\frac{d}{r}$ is obviously satisfied for $z>0$. Taking the above relationship back to the dual transform definition, $\tilde{U}(z)$ is directly acquired after elementary calculation,
	\begin{equation*}
	\tilde{U}(z)=\frac{1-\delta(1-k)}{\delta(1-k)}z^{\frac{\delta(1-k)}{\delta(1-k)-1}}K_{1}\bar{L}^{\frac{(1-k)(1-\delta)}{1-\delta(1-k)}}-\frac{d}{r}z.
	\end{equation*}
	
	\noindent (2) $R_{post}>\frac{d}{r}$: Considering the fact $V_{\scriptscriptstyle PR}(x)\!=\!\inf\limits_{\lambda_{\scriptscriptstyle PR}>0}[v_{\scriptscriptstyle PR}(\lambda_{\scriptscriptstyle PR})\!+\!\lambda_{\scriptscriptstyle PR} x]$, it can be obtained that $v_{\scriptscriptstyle PR}(z)$ is the Legendre-Fenchel transform of $V_{\scriptscriptstyle PR}(x)$ from \cite[Chapter 3, Lemma 4.3]{karatzas1998methods}. Then the identical forms of functions $V_{\scriptscriptstyle PR}(x)$ and $U(x)$ enable us to deduce the solution as $\tilde{U}(z)=v_{\scriptscriptstyle PR}(z)$. The last step is to claim $x^{*}\ge R_{post}$, which can be resorted to the condition $x^{*}=-v_{\scriptscriptstyle PR}^{\prime}(z)\ge -v_{\scriptscriptstyle PR}^{\prime}(\hat{z}_{\scriptscriptstyle PR})=R_{post}$.
		
	\section{Proof of Lemma \ref{Lemma 4}}\label{Appendix 4}	
	\noindent Following \cite[Section 6, Lemma 6.3]{karatzas2000utility}, we first define a new continuous process as
		\begin{equation*}
			G(t)\triangleq\frac{1}{\xi(t)}\tilde{\mathbb{E}}\left[\left.\int_{t}^{\tau}\xi(s)(c(s)+wl(s)+d-w\bar{L})d s+\xi(\tau)K\right|\mathcal{F}_{t}\right],\quad\forall t\in[0,\tau],
		\end{equation*}
		where $\tilde{\mathbb{E}}[\cdot]$ representing the expectation under $\tilde{\mathbb{P}}$ measure. From the property of the random variable $K$, this process satisfies
		$G(t)\!=\!\frac{1}{\xi(t)}\tilde{\mathbb{E}}\!\left[\!\left.\int_{t}^{\tau}\!\xi(s)(c(s)\!+\!w l(s))d s\!+\!\xi(\tau)\left(K\!-\!\frac{d-w\bar{L}}{r}\right)\right|\mathcal{F}_{t}\!\right]\!+\!\frac{d\!-\!w\bar{L}}{r}\!\ge\!\frac{d\!-\!w\bar{L}}{r}$, a.s..
		Then, making use of the condition $\mathbb{E}\left[\int_{0}^{\tau}H(t)(c(t)+w l(t)+d-w\bar{L})d t+H(\tau )K\right]=x$, we get $G(\tau)=K$ and
		\begin{equation*}
			\begin{split}
				G(0)&=\tilde{\mathbb{E}}\left[\int_{0}^{\tau}\xi(s)(c(s)+w l(s)+d-w\bar{L})d s+\xi(\tau) K\right]\\
				&=\mathbb{E}\left[\int_{0}^{\tau}H(s)(c(s)+w l(s)+d-w\bar{L})d s+H(\tau) K\right]=x,
			\end{split}
		\end{equation*}
		the above derivation involves changing the measure from the $\tilde{\mathbb{P}}$ measure with the pricing kernel as $\xi(t)$ to the $\mathbb{P}$ measure with the pricing kernel as $H(t)$.
		Meanwhile, we define a new process
		\begin{equation*}
			M(t)=\xi(t)G(t)+\int_{0}^{t}\xi(s)(c(s)+w l(s)+d-w\bar{L})d s,\quad \forall t\in[0,\tau].
		\end{equation*}
		Based on the fact
		\begin{equation*}
			\begin{split}
				\tilde{\mathbb{E}}[M(t)]&=\!\tilde{\mathbb{E}}\left[\tilde{\mathbb{E}}\left[\left.\int_{t}^{\tau}\xi(s)(c(s)\!+\!w l(s)\!+\!d\!-\!w\bar{L})d s\!+\!\xi(\tau)K\right|\!\mathcal{F}_{t}\right]\!+\!\int_{0}^{t}\xi(s)(c(s)\!+\!w l(s)\!+\!d\!-\!w\bar{L})d s\right]\\
				&=\tilde{\mathbb{E}}\left[\int_{0}^{\tau}\xi(s)(c(s)+w l(s)+d-w\bar{L})d s+\xi(\tau)K\right]=x=M(0),
			\end{split}
		\end{equation*}
		$M(t)$ is a $\tilde{\mathbb{P}}$-martingale. According to the Martingale Representation Theorem from Bjork (2009), Chapter 11, Theorem 11.2, it can be expressed as $M(t)=x+\int_{0}^{t}\rho(s)d\tilde{B}(s)$, $\forall t \in[0,\tau]$, with an $\mathbb{F}$-adapted process $\rho(t)$ satisfying $\int_{0}^{\infty}\rho^{2}(s)d s<\infty$ a.s.. Furthermore, adopting the portfolio strategy $\pi(t)\triangleq\frac{\rho(t)}{\sigma \xi(t)}$, the wealth process becomes
		\begin{equation*}
			\begin{split}
				d X^{x,c,\pi,l}(t)&=r X^{x,c,\pi,l}(t)d t+\pi(t)(\mu-r)d t-(c(t)+w l(t)+d-w\bar{L})d t+\sigma\pi(t)dB(t)\\
				&=r X^{x,c,\pi,l}(t)d t-(c(t)+w l(t)+d-w\bar{L})d t+\sigma\pi(t)d\tilde{B}(t)\\
				&=r X^{x,c,\pi,l}(t)d t-(c(t)+w l(t)+d-w\bar{L})d t+\frac{\rho(t)}{\xi(t)}d\tilde{B}(t),
			\end{split}
		\end{equation*}
		the second equality also comes from changing the measure by $\tilde{B}(t)\triangleq B(t)+\theta t$. We can observe that $G(t)=X^{x,c,\pi,l}(t)$ a.s. on $[0,\tau]$, which concludes the proof of this lemma. 
		
	\section{Proof of Lemma \ref{Lemma 5}}\label{Appendix 5}

\begin{remark}
In this section we prove Lemma \ref{Lemma 5}. Moreover, we also show that the conditions $\bar{z}\!<\!\tilde{y},$ defined below, and $\bar{z}\!<\!\hat{z}_{\scriptscriptstyle PR}$ hold true.
\end{remark}
	
	\noindent The proof here refers to Oksendal (2013) Section 10, Example 10.3.1.\footnote{B. Oksendal, Stochastic differential equations: an introduction with applications. Springer Science \& Business Media, 2013.} First of all, \cite[Lemma 2.1]{2021arXiv210615426D} shows that $$\tilde{u}(z)\!=\!\left[A_{1}z^{\frac{\delta(1\!-\!k)}{\delta(1\!-\!k)\!-\!1}}\!-\!w L z\right]\mathbb{I}_{\{0<z<\tilde{y}\}}\!+\!\left[A_{2}z^{\!-\!\frac{1\!-\!k}{k}}\right]\mathbb{I}_{\{z\ge\tilde{y}\}},$$
	 with $A_{1}\!\triangleq\!\frac{1\!-\!\delta\!+\!\delta k}{\delta (1\!-\!k)}L^{\frac{(1\!-\!k)(1\!-\!\delta)}{1\!-\!\delta(1\!-\!k)}}$, $A_{2}\triangleq\frac{k}{\delta(1\!-\!k)}\left(\frac{1\!-\!\delta}{\delta w}\right)^{\frac{(1\!-\!k)(1\!-\!\delta)}{k}}$, and  $\tilde{y}\triangleq L^{\!-k}\left(\frac{1\!-\!\delta}{\delta w}\right)^{1\!-\!\delta(1\!-\!k)}$.

 Introducing two functions
	\begin{equation*}
	    g(t,z)\triangleq e^{-\gamma t}\tilde{U}(z), \quad G(t,z,\bar{w})\triangleq g(t,z)+\bar{w}=e^{-\gamma t}\tilde{U}(z)+\bar{w},
	\end{equation*}
	and an operator $\mathcal{A}_{P}G(t,z,\bar{w})\triangleq\frac{\partial G}{\partial t}+(\gamma-r)z\frac{\partial G}{\partial z}+\frac{\theta^{2}}{2}z^{2}\frac{\partial^{2}G}{\partial z^{2}}+e^{-\gamma t}\tilde{u}(z)-e^{-\gamma t}(d-w\bar{L})z$, we can determine the continuous region as
	$\Omega_{1}=\{(t,z,\bar{w}):\mathcal{A}_{P}G(t,z,\bar{w})>0\}$. Moreover, since
	\begin{equation*}
	\mathcal{A}_{P}G(t,z,\bar{w})=-\gamma e^{-\gamma t}\tilde{U}(z)+(\gamma-r)z e^{-\gamma t}\tilde{U}^{\prime}(z)+\frac{\theta^{2}}{2}z^{2} e^{-\gamma t}\tilde{U}^{\prime\prime}(z)+e^{-\gamma t}\left(\tilde{u}(z)-(d-w\bar{L})z\right),
	\end{equation*}
	defining a new function $h(z)=-\gamma\tilde{U}(z)+(\gamma-r)z \tilde{U}^{\prime}(z)+\frac{\theta^{2}}{2}z^{2} \tilde{U}^{\prime\prime}(z)+\tilde{u}(z)-(d-w\bar{L})z$, the continuous region can be rewritten as $\Omega_{1}=\{z>0:h(z)>0\}$. Since the function $\tilde{U}(z)$ takes two different forms based on the value of $R_{post}$, we split the remaining discussion also into two cases: $R_{post}=\frac{d}{r}$ and $R_{post}>\frac{d}{r}$.

	\noindent (1) For $R_{post}\!=\!\frac{d}{r}$, we have $\tilde{U}(z)\!=\!\frac{1\!-\!\delta(1\!-\!k)}{\delta(1\!-\!k)}K_{1}\bar{L}^{\frac{(1\!-\!k)(1\!-\!\delta)}{1\!-\!\delta(1\!-\!k)}}\!z^{\frac{\delta(1\!-\!k)}{\delta(1\!-\!k)\!-\!1}}\!-\!\frac{d}{r}z$. After the basic calculation, we get
	\begin{equation}
	\begin{split}
	h(z)&=\frac{\delta(1-k)-1}{\delta(1-k)}\bar{L}^{\frac{(1-k)(1-\delta)}{1-\delta(1-k)}}z^{\frac{\delta(1-k)}{\delta(1-k)-1}}+w\bar{L}z+\tilde{u}(z)\\
	&=\frac{\delta(1\!-\!k)\!-\!1}{\delta(1\!-\!k)}\bar{L}^{\frac{(1\!-\!k)(1\!-\!\delta)}{1\!-\!\delta(1\!-\!k)}}z^{\frac{\delta(1\!-\!k)}{\delta(1\!-\!k)\!-\!1}}\!+\!w\bar{L}z\!+\!\left[A_{1}z^{\frac{\delta (1\!-\!k)}{\delta(1\!-\!k)\!-\!1}}\!-\!w L z\right]\!\mathbb{I}_{\{0<z<\tilde{y}\}}\!+\!\left[A_{2}z^{-\frac{1\!-\!k}{k}}\right]\!\mathbb{I}_{\{z\ge\tilde{y}\}}.
	\end{split}\label{Equation Appendix 2}
	\end{equation}
	$h(z)$ inherits the piecewise form from the function $\tilde{u}(z)$. Afterwards, determining the continuous region corresponds to characterize the features of the zero of $h(z)$. We begin claiming its convexity by the second derivative function. On the interval $0<z<\tilde{y}$, we can directly determine the sign of $h^{\prime \prime}(z)$ with
	$h^{\prime \prime}(z)=\frac{1}{\delta(1-k)-1}z^{\frac{2-\delta(1-k)}{\delta(1-k)-1}}\left[\bar{L}^{\frac{(1-k)(1-\delta)}{1-\delta(1-k)}}-L^{\frac{(1-k)(1-\delta)}{1-\delta(1-k)}}\right]>0$. As for the interval $z>\tilde{y}$, the corresponding second derivative function $h^{\prime\prime}(z)$ is
	\begin{equation*}
	h^{\prime\prime}(z)=\frac{1}{\delta(1-k)-1}\bar{L}^{\frac{(1-k)(1-\delta)}{1-\delta(1-k)}}z^{\frac{2-\delta(1-k)}{\delta(1-k)-1}}+\frac{1}{\delta k}\left(\frac{1-\delta}{\delta w}\right)^{\frac{(1-k)(1-\delta)}{k}}z^{-\frac{1+k}{k}}.
	\end{equation*}
	By solving the inequality, $\frac{1}{\delta(1-k)-1}\bar{L}^{\frac{(1-k)(1-\delta)}{1-\delta(1-k)}}z^{\frac{2-\delta(1-k)}{\delta(1-k)-1}}+\frac{1}{\delta k}\left(\frac{1-\delta}{\delta w}\right)^{\frac{(1-k)(1-\delta)}{k}}z^{-\frac{1+k}{k}}>0$, we get
	\begin{equation*}
	z>\bar{L}^{-k}\left(\frac{1-\delta(1-k)}{\delta k}\right)^{\frac{k(\delta(1-k)-1)}{(1-k)(\delta-1)}}\left(\frac{1-\delta}{\delta w}\right)^{1-\delta(1-k)}.
	\end{equation*} 	
	Since
	\begin{equation*}
	\tilde{y}-\bar{L}^{\!-\!k}\left(\!\frac{1\!-\!\delta(1\!-\!k)}{\delta k}\!\right)^{\frac{k(\delta(1\!-\!k)\!-\!1)}{(1\!-\!k)(\delta\!-\!1)}}\!\!\left(\!\frac{1\!-\!\delta}{\delta w}\!\right)^{1\!-\!\delta(1\!-\!k)}\!\!\!=\!\left(\!\frac{1\!-\!\delta}{\delta w}\!\right)^{1\!-\!\delta(1\!-\!k)}\!\!\left[\!L^{\!-\!k}\!-\!\left(\!\frac{1\!-\!\delta(1\!-\!k)}{\delta k}\!\right)^{\frac{k(\delta(1\!-\!k)\!-\!1)}{(1\!-\!k)(\delta\!-\!1)}}\!\!\bar{L}^{\!-\!k}\!\right]\!\!>\!0,
	\end{equation*}	
	$h^{\prime\prime}(z)>0$ keeps true for $z>\tilde{y}$. Besides, considering the condition $\lim\limits_{z\uparrow \tilde{y}}h^{\prime\prime}(z)=\lim\limits_{z\downarrow \tilde{y}}h^{\prime\prime}(z)$, we can conclude that the function $h(z)$ is strictly convex on the interval $z>0$. Then we move to claim $h(\tilde{y})>0$: before this, a new function is introduced as
	\begin{equation*}
	f(z)=\frac{\delta(1-k)-1}{\delta(1-k)}\bar{L}^{\frac{(1-k)(1-\delta)}{1-\delta(1-k)}}z^{\frac{\delta(1-k)}{\delta(1-k)-1}}+\frac{k}{\delta(1-k)}\left(\frac{1-\delta}{\delta w}\right)^{\frac{(1-k)(1-\delta)}{k}}\!\!\!z^{-\frac{1-k}{k}}+w\bar{L}z, \quad z>0,
	\end{equation*}	
	and its derivative functions are
	\begin{equation*}
	f^{\prime}(z)=z^{\frac{1}{\delta(1-k)-1}}\bar{L}^{\frac{(1-k)(1-\delta)}{1-\delta(1-k)}}-\frac{1}{\delta}\left(\frac{1-\delta}{\delta w}\right)^{\frac{(1-k)(1-\delta)}{k}}z^{-\frac{1}{k}}+w\bar{L},
	\end{equation*}
	\begin{equation*}
	f^{\prime\prime}(z)=\frac{1}{\delta(1-k)-1}z^{\frac{2-\delta(1-k)}{\delta(1-k)-1}}\bar{L}^{\frac{(1-k)(1-\delta)}{1-\delta(1-k)}}+\frac{1}{\delta k}\left(\frac{1-\delta}{\delta w}\right)^{\frac{(1-k)(1-\delta)}{k}}z^{-\frac{1+k}{k}}.
	\end{equation*}
	Defining $\tilde{y}_{1}=\bar{L}^{-k}\left(\frac{1-\delta}{\delta w}\right)^{1-\delta(1-k)}$, it can be obtained that
	\begin{equation*}
	f(\tilde{y}_{1})\!=\!\frac{\delta(1\!-\!k)\!-\!1}{\delta(1\!-\!k)}\bar{L}^{1\!-\!k}\!\left(\frac{1\!-\!\delta}{\delta w}\right)^{\!-\!\delta(1\!-\!k)}\!\!+\!\frac{k}{\delta(1\!-\!k)}\bar{L}^{1\!-\!k}\!\left(\frac{1\!-\!\delta}{\delta w}\right)^{\!-\!\delta(1\!-\!k)}\!\!+\!w\bar{L}^{1\!-\!k}\left(\frac{1\!-\!\delta}{\delta w}\right)^{1\!-\!\delta(1\!-\!k)}\!\!\!=\!0,
	\end{equation*}
	\begin{equation*}
	f^{\prime}(\tilde{y}_{1})\!=\!\!\left(\!\bar{L}^{\!-\!k}\!\left(\!\frac{1\!-\!\delta}{\delta w}\!\right)^{1\!-\!\delta(1\!-\!k)}\!\right)^{\frac{1}{\delta(1\!-\!k)\!-\!1}}\!\!\!\bar{L}^{\frac{(1\!-\!k)(1\!-\!\delta)}{1\!-\!\delta(1\!-\!k)}}\!\!-\!\frac{1}{\delta}\!\left(\!\frac{1\!-\!\delta}{\delta w}\!\right)^{\frac{(1\!-\!k)(1\!-\!\delta)}{k}}\!\!\!\left(\!\bar{L}^{\!-\!k}\!\left(\!\frac{1\!-\!\delta}{\delta w}\!\right)^{1\!-\!\delta(1\!-\!k)}\!\right)^{\!-\!\frac{1}{k}}\!+\!w\bar{L}\!\!=\!0.
	\end{equation*}
	For the second derivative, $f^{\prime\prime}(z)\!>\!0$ is equivalent to $z\!>\!\left(\frac{1\!-\!\delta(1\!-\!k)}{\delta k}\right)^{\frac{k(\delta(1-k)-1)}{(1-k)(\delta-1)}}\left(\frac{1\!-\!\delta}{\delta w}\right)^{1\!-\!\delta(1\!-\!k)}\bar{L}^{\!-\!k}$. Since $0\!<\!\left(\frac{1-\delta(1-k)}{\delta k}\right)^{\frac{k(\delta(1-k)-1)}{(1-k)(\delta-1)}}\!\!<\!1$, we have $\tilde{y}_{1}\!>\!\left(\frac{1\!-\!\delta(1\!-\!k)}{\delta k}\right)^{\frac{k(\delta(1-k)-1)}{(1-k)(\delta-1)}}\left(\frac{1\!-\!\delta}{\delta w}\right)^{1\!-\!\delta(1\!-\!k)}\bar{L}^{-k}$, which results in $f^{\prime\prime}(z)>0$ for $z>\tilde{y}_{1}$. Then the fact $f^{\prime}(\tilde{y}_{1})=0$ indicates that $f^{\prime}(z)>0$ for $z>\tilde{y}_{1}$, which means $f(z)$ is strictly increasing on the corresponding interval. Considering the relationship $\tilde{y}=L^{-k}\left(\frac{1-\delta}{\delta w}\right)^{1-\delta(1-k)}>\tilde{y}_{1}$, we can observe the positive value of $h(\tilde{y})$ through
	\begin{equation}
	0=f(\tilde{y}_{1})<f(\tilde{y})=h(\tilde{y}).\label{Equation Appendix 4}
	\end{equation}
	Finally, in view of the limitations
	\begin{equation*}
	\lim\limits_{z\downarrow 0}h(z)=\lim\limits_{z\downarrow 0}\left[\frac{\delta(1-k)-1}{\delta(1-k)}z^{\frac{\delta(1-k)}{\delta(1-k)-1}}\left(\bar{L}^{\frac{(1-k)(1-\delta)}{1-\delta(1-k)}}-L^{\frac{(1-k)(1-\delta)}{1-\delta(1-k)}}\right)+w z(\bar{L}-L)\right]=0,
	\end{equation*}
	\begin{equation*}
	\lim\limits_{z\downarrow 0}h^{\prime}(z)=\lim\limits_{z\downarrow 0}\left[z^{\frac{1}{\delta(1-k)-1}}\left(\bar{L}^{\frac{(1-k)(1-\delta)}{1-\delta(1-k)}}-L^{\frac{(1-k)(1-\delta)}{1-\delta(1-k)}}\right)+w(\bar{L}-L)\right]=-\infty,
	\end{equation*}
	and the properties $h^{\prime\prime}(z)>0$ for $z>0$, $h(\tilde{y})>0$, we can conclude that there is a unique zero of $h(z)$, which is denoted as $\bar{z}$, satisfying $\bar{z}<\tilde{y}$ and
	$h^{\prime}(\bar{z})=\bar{L}^{\frac{(1-k)(1-\delta)}{1-\delta(1-k)}}\bar{z}^{\frac{1}{\delta(1-k)-1}}+w\bar{L}+\tilde{u}^{\prime}(\bar{z})$. Hence, the continuous region is
	$\Omega_{1}=\{(t,z,\bar{w}):\mathcal{A}_{P}G(s,z,\bar{w})>0\}=\{z:h(z)>0\}=\{z>\bar{z}\}$.
		
	\noindent (2) For $R_{post}>\frac{d}{r}$, we have
	\begin{equation*}
	\tilde{U}(z)=\begin{cases}
	B_{2,\scriptscriptstyle PR}\hat{z}_{\scriptscriptstyle PR}^{n_{2}}+\frac{1-\delta(1-k)}{\delta(1-k)}K_{1}\bar{L}^{\frac{(1-k)(1-\delta)}{1-\delta(1-k)}}\hat{z}_{\scriptscriptstyle PR}^{\frac{\delta(1-k)}{\delta(1-k)-1}}-\frac{d}{r}\hat{z}_{\scriptscriptstyle PR}-R_{post}(z-\hat{z}_{\scriptscriptstyle PR}), & z\!\ge\!\hat{z}_{\scriptscriptstyle PR},\\
	B_{2,\scriptscriptstyle PR}z^{n_{2}}+\frac{1-\delta(1-k)}{\delta(1-k)}K_{1}\bar{L}^{\frac{(1-k)(1-\delta)}{1-\delta(1-k)}}z^{\frac{\delta(1-k)}{\delta(1-k)-1}}-\frac{d}{r}z, & 0\!<\!z\!<\!\hat{z}_{\scriptscriptstyle PR},
	\end{cases}
	\end{equation*} 
	then the function $h(z)$ on the interval $0<z<\hat{z}_{\scriptscriptstyle PR}$ is obtained as
	\begin{equation*}
	h(z)\!=\!B_{2,\scriptscriptstyle PR}z^{n_{2}}\left[-\!\gamma\!+\!(\gamma\!-\!r)n_{2}\!+\!\frac{\theta^{2}}{2}n_{2}(n_{2}\!-\!1)\right]+\frac{\delta(1-k)-1}{\delta(1-k)}\bar{L}^{\frac{(1-k)(1-\delta)}{1-\delta(1-k)}}z^{\frac{\delta(1-k)}{\delta(1-k)-1}}+w\bar{L}z+\tilde{u}(z).
	\end{equation*}
	From $n_{1}+n_{2}=-\frac{\gamma-r-\frac{\theta^{2}}{2}}{\frac{\theta^{2}}{2}}$ and $n_{1}n_{2}=-\frac{\gamma}{\frac{\theta^{2}}{2}}$, we can deduce $-\gamma+(\gamma-r)n_{2}+\frac{\theta^{2}}{2}n_{2}(n_{2}-1)=0$. Hence, the function $h(z)$ is reduced as
	\begin{equation*}
	h(z)=\frac{\delta(1-k)-1}{\delta(1-k)}\bar{L}^{\frac{(1-k)(1-\delta)}{1-\delta(1-k)}}z^{\frac{\delta(1-k)}{\delta(1-k)-1}}+w\bar{L}z+\tilde{u}(z),\quad 0<z<\hat{z}_{\scriptscriptstyle PR}.
	\end{equation*}
	Compared to Equation (\ref{Equation Appendix 2}), we can observe that $h(z)$ adopts the same form but applies to the different intervals. As for the interval $z\ge\hat{z}_{\scriptscriptstyle PR}$, considering that
	$\tilde{U}(z)=v_{\scriptscriptstyle PR}(\hat{z}_{\scriptscriptstyle PR})-R_{post}(z-\hat{z}_{\scriptscriptstyle PR})$,
	and the condition (\ref{Equation Appendix 3}) is applicable at the point $z=\hat{z}_{\scriptscriptstyle PR}$, we get
	\begin{equation*}
	\begin{split}
	h(z)&=-\gamma \left(v_{\scriptscriptstyle PR}(\hat{z}_{\scriptscriptstyle PR})-R_{post}(z-\hat{z}_{\scriptscriptstyle PR})\right)-(\gamma -r)z R_{post}+\tilde{u}(z)-(d-w\bar{L})z\\
	&=(r R_{post}-d)(z-\hat{z}_{\scriptscriptstyle PR})-\tilde{u}_{\scriptscriptstyle PR}(\hat{z}_{\scriptscriptstyle PR})+\tilde{u}(z)+w\bar{L}z.
	\end{split}
	\end{equation*}
	$h^{\prime\prime}(z)\!=\!\tilde{u}^{\prime\prime}(z)\!>\!0$ shows that $h(z)$ is strictly convex on $(\hat{z}_{\scriptscriptstyle PR},\infty)$. Then a contradiction is constructed to prove $\bar{z}\!<\!\hat{z}_{\scriptscriptstyle PR}$. We first use $r\!=\!\frac{\theta^{2}}{2}(1\!-\!n_{2})(n_{1}\!-\!1)$ and $n_{2}\!>\!1\!>\!\frac{\delta(1\!-\!k)}{\delta(1\!-\!k)\!-\!1}\!>\!0$ to derive the condition
	\begin{equation*}
	\frac{(1-n_{2})(1-\delta(1-k))}{n_{2}(\delta(1-k)-1)-\delta(1-k)}\frac{1}{K_{1}}=r\frac{\gamma-r\delta(1-k)-\frac{\theta^{2}}{2}\frac{\delta(1-k)}{1-\delta(1-k)}}{\gamma-r\delta(1-k)+\frac{\theta^{2}}{2}n_{2}}<r,
	\end{equation*}
	which gives us 
	\begin{equation*}
	\hat{z}_{\scriptscriptstyle PR}^{\frac{1}{\delta(1-k)-1}}=\bar{L}^{\frac{(1-k)(1-\delta)}{\delta(1-k)-1}}\frac{(1-n_{2})(1-\delta(1-k))}{n_{2}(\delta(1-k)-1)-\delta(1-k)}\frac{\left(R_{post}-\frac{d}{r}\right)}{K_{1}}<\bar{L}^{\frac{(1-k)(1-\delta)}{\delta(1-k)-1}}\left(r R_{post}-d\right).
	\end{equation*}
	Assuming $\bar{z}\ge\hat{z}_{\scriptscriptstyle PR}$, we can observe the contradiction through
	\begin{equation*}
	\begin{split}
	h^{\prime}(\bar{z})&=\bar{L}^{\frac{(1-k)(1-\delta)}{1-\delta(1-k)}}\bar{z}^{\frac{1}{\delta(1-k)-1}}+w\bar{L}+\tilde{u}^{\prime}(\bar{z})\\
	&\leq \bar{L}^{\frac{(1-k)(1-\delta)}{1-\delta(1-k)}}\hat{z}_{\scriptscriptstyle PR}^{\frac{1}{\delta(1-k)-1}}+w\bar{L}+\tilde{u}^{\prime}(\bar{z})\\
	&<\bar{L}^{\frac{(1-k)(1-\delta)}{1-\delta(1-K)}}\bar{L}^{\frac{(1-k)(1-\delta)}{\delta(1-K)-1}}(r R_{post}-d)+w\bar{L}+\tilde{u}^{\prime}(\bar{z})\\
	&=(r R_{post}-d)+w\bar{L}+\tilde{u}^{\prime}(\bar{z})=h^{\prime}(\bar{z}).
	\end{split}
	\end{equation*}
	Then the condition $\bar{z}<\hat{z}_{\scriptscriptstyle PR}$ implies $h(\hat{z}_{\scriptscriptstyle PR})> 0$ and
	\begin{equation*}
	\lim\limits_{z\uparrow\hat{z}_{\scriptscriptstyle PR}}h^{\prime}(z)=\lim\limits_{z\uparrow\hat{z}_{\scriptscriptstyle PR}}\left[\bar{L}^{\frac{(1-k)(1-\delta)}{1-\delta(1-k)}}z^{\frac{1}{\delta(1-k)-1}}+\tilde{u}^{\prime}(z)+w\bar{L}\right]=-\tilde{u}_{\scriptscriptstyle PR}^{\prime}(\hat{z}_{\scriptscriptstyle PR})+\tilde{u}^{\prime}(\hat{z}_{\scriptscriptstyle PR})+w\bar{L}\ge 0.
	\end{equation*}
	Afterwards, we have
	\begin{equation*}
	\begin{split}
	\lim\limits_{z\downarrow \hat{z}_{\scriptscriptstyle PR}}h^{\prime}(z)&=r R_{post}-d+\tilde{u}^{\prime}(\hat{z}_{\scriptscriptstyle PR})+w\bar{L}\\
	&\ge r R_{post}-d+\tilde{u}_{\scriptscriptstyle PR}^{\prime}(\hat{z}_{\scriptscriptstyle PR})\\
	&=r R_{post}-d-\hat{z}_{\scriptscriptstyle PR}^{\frac{1}{\delta(1-k)-1}}\bar{L}^{\frac{(1-k)(1-\delta)}{1-\delta(1-k)}}\\
	&=r R_{post}-d-\frac{(1-n_{2})(1-\delta(1-k))}{n_{2}(\delta(1-k)-1)-\delta(1-k)}\frac{1}{K_{1}}\left(R_{post}-\frac{d}{r}\right)>0,
	\end{split}
	\end{equation*}
	which indicates that $h(z)$ is strictly increasing for $z>\hat{z}_{\scriptscriptstyle PR}$ regarding the convex property already shown. Therefore, $\bar{z}$ is the unique zero of function $h(z)$, and satisfies $\bar{z}<\hat{z}_{\scriptscriptstyle PR}$. The last step is to claim $\bar{z}<\tilde{y}$ under this case, which is equivalent to $h(\tilde{y})>0$ and discussed in two different situations. If $\tilde{y}\leq\hat{z}_{\scriptscriptstyle PR}$, using the result (\ref{Equation Appendix 4}), we have
	\begin{equation*}
	h(\tilde{y})=\frac{\delta(1-k)-1}{\delta(1-k)}\bar{L}^{\frac{(1-k)(1-\delta)}{1-\delta(1-k)}}\tilde{y}^{\frac{\delta(1-k)}{\delta(1-k)-1}}+w\bar{L}\tilde{y}+\tilde{u}(\tilde{y})=f(\tilde{y})>0,
	\end{equation*}
	otherwise, if $\tilde{y}\!>\!\hat{z}_{\scriptscriptstyle PR}$, we use the increasing property of $h(z)$ on $(\hat{z}_{\scriptscriptstyle PR},\infty)$ to directly obtain $h(\tilde{y})\!>\!0$.

	\section{Calculation of Variational Inequalities (\ref{Equation 9})}\label{Appendix 7}
	
	\noindent The solution of (\ref{Equation 9}) is split into two different cases based on the value of $R_{post}$, namely $R_{post}\!=\!\frac{d}{r}$ and  $R_{post}\!>\!\frac{d}{r}$. Following \cite[Appendix A]{choi2008optimal}, we take the time-separated form of function ${\phi}(t,z)\!=\!e^{\!-\!\gamma t}v(z)$ for solving the above variational inequalities explicitly. 

We recall that \cite[Lemma 2.1]{2021arXiv210615426D} shows that $\tilde{u}(z)\!=\!\left[A_{1}z^{\frac{\delta(1\!-\!k)}{\delta(1\!-\!k)\!-\!1}}\!-\!w L z\right]\mathbb{I}_{\{0<z<\tilde{y}\}}\!+\!\left[A_{2}z^{\!-\!\frac{1\!-\!k}{k}}\right]\mathbb{I}_{\{z\ge\tilde{y}\}}$,
	 with $A_{1}\!\triangleq\!\frac{1\!-\!\delta\!+\!\delta k}{\delta (1\!-\!k)}L^{\frac{(1\!-\!k)(1\!-\!\delta)}{1\!-\!\delta(1\!-\!k)}}$, $A_{2}\triangleq\frac{k}{\delta(1\!-\!k)}\left(\frac{1\!-\!\delta}{\delta w}\right)^{\frac{(1\!-\!k)(1\!-\!\delta)}{k}}$, and  $\tilde{y}\triangleq L^{\!-k}\left(\frac{1\!-\!\delta}{\delta w}\right)^{1\!-\!\delta(1\!-\!k)}$. Moreover, $n_{1}$ and $n_{2}$ are the roots of the second-order equation
		$\frac{\theta^{2}}{2}n^{2}+\left(\gamma-r-\frac{\theta^{2}}{2}\right)n-\gamma=0$,
		and satisfy 
		\begin{equation*}
		n_{1,2}=-\frac{\gamma-r-\frac{\theta^{2}}{2}}{\theta^{2}}\mp\sqrt{\left(\frac{\gamma-r-\frac{\theta^{2}}{2}}{\theta^{2}}\right)^{2}+\frac{2\gamma}{\theta^{2}}},\qquad n_{1}<0,\quad n_{2}>1.
		\end{equation*}

	\paragraph{Case 1. $R_{pre}=\frac{d-w\bar{L}}{r}$ $\&$ $R_{post}=\frac{d}{r}$}~{}
	\newline
	From the condition $(V1)$ of (\ref{Equation 9}), the following differential equation holds in the region $z>\bar{z}$,
	\begin{equation}
		-\gamma v(z)+(\gamma-r)z v^{\prime}(z)+\frac{1}{2}\theta^{2}z^{2}v^{\prime\prime}(z)+\tilde{u}(z)-(d-w\bar{L})z=0,\label{Equation Appendix 5}
	\end{equation}
	whose solution takes a form as
	\begin{equation*}
		v(z)=\begin{cases}
			B_{11}z^{n_{1}}+B_{21}z^{n_{2}}+\frac{A_{1}}{\Gamma_{1}}z^{\frac{\delta(1-k)}{\delta(1-k)-1}}+\frac{w(\bar{L}-L)-d}{r}z, & \bar{z}<z<\tilde{y},\\
			
			B_{12}z^{n_{1}}+B_{22}z^{n_{2}}+\frac{A_{2}}{\Gamma_{2}}z^{-\frac{1-k}{k}}+\frac{w\bar{L}-d}{r}z,& z\ge\tilde{y}.
		\end{cases}
	\end{equation*}
	Since $n_{2}>0$, for the sake of avoiding the explosion of the term $z^{n_{2}}$ as $z$ goes to $\infty$, we set $B_{22}=0$.
	Then, the condition $(V4)$ of (\ref{Equation 9}) enables us to obtain
	\begin{equation*}
		v(\bar{z})=\tilde{U}(\bar{z})=\frac{1-\delta(1-k)}{\delta(1-k)}K_{1}\bar{L}^{\frac{(1-k)(1-\delta)}{1-\delta(1-k)}}\bar{z}^{\frac{\delta(1-k)}{\delta(1-k)-1}}-\frac{d}{r}\bar{z},
	\end{equation*} 
	the second equality results from the condition $R_{post}=\frac{d}{r}$. Furthermore, combining with the smooth condition at the point $z=\tilde{y}$, we can construct a four-equations system to determine the parameters $B_{11}$, $B_{21}$, $B_{12}$ and $\bar{z}$.
	\begin{itemize}
		\item $\mathcal{C}^{0}$ condition at $z=\bar{z}$
		\begin{equation*}
			B_{11}\bar{z}^{n_{1}}+B_{21}\bar{z}^{n_{2}}+\frac{A_{1}}{\Gamma_{1}}\bar{z}^{\frac{\delta(1-k)}{\delta(1-k)-1}}+\frac{w(\bar{L}-L)}{r}\bar{z}=\frac{1-\delta(1-k)}{\delta(1-k)}K_{1}\bar{L}^{\frac{(1-k)(1-\delta)}{1-\delta(1-k)}}\bar{z}^{\frac{\delta(1-k)}{\delta(1-k)-1}};
		\end{equation*}
		\item $\mathcal{C}^{1}$ condition at $z=\bar{z}$
		\begin{equation*}
			n_{1}B_{11}\bar{z}^{n_{1}\!-\!1}\!+\!n_{2}B_{21}\bar{z}^{n_{2}\!-\!1}\!+\!\frac{\delta(1\!-\!k)}{\delta(1\!-\!k)\!-\!1}\frac{A_{1}}{\Gamma_{1}}\bar{z}^{\frac{1}{\delta(1\!-\!k)\!-\!1}}\!+\!\frac{w(\bar{L}\!-\!L)}{r}=\!-\!K_{1}\bar{L}^{\frac{(1\!-\!k)(1\!-\!\delta)}{1\!-\!\delta(1\!-\!k)}}\bar{z}^{\frac{1}{\delta(1\!-\!k)\!-\!1}};
		\end{equation*}
		\item $\mathcal{C}^{0}$ condition at $z=\tilde{y}$
		\begin{equation*}
			B_{11}\tilde{y}^{n_{1}}+B_{21}\tilde{y}^{n_{2}}+\frac{A_{1}}{\Gamma_{1}}\tilde{y}^{\frac{\delta(1-k)}{\delta(1-k)-1}}-\frac{w L}{r}\tilde{y}=B_{12}\tilde{y}^{n_{1}}+\frac{A_{2}}{\Gamma_{2}}\tilde{y}^{-\frac{1-k}{k}};
		\end{equation*}
		\item $\mathcal{C}^{1}$ condition at $z=\tilde{y}$
		\begin{equation*}
			n_{1}B_{11}\tilde{y}^{n_{1}-1}+n_{2}B_{21}\tilde{y}^{n_{2}-1}+\frac{\delta(1-k)}{\delta(1-k)-1}\frac{A_{1}}{\Gamma_{1}}\tilde{y}^{\frac{1}{\delta(1-k)-1}}-\frac{w L}{r}= n_{1}B_{12}\tilde{y}^{n_{1}-1}-\frac{1-k}{k}\frac{A_{2}}{\Gamma_{2}}\tilde{y}^{-\frac{1}{k}}.
		\end{equation*}
	\end{itemize}
	
	\paragraph{Case 2. $R_{pre}=\frac{d-w\bar{L}}{r}$ $\&$ $R_{post}>\frac{d}{r}$}~{}
	\newline
	Then we move to the second case with a different condition $R_{post}>\frac{d}{r}$ compared to Case 1, which mainly affects the post-retirement part and leads to a different form of $\tilde{U}(z)$. Lemma \ref{Lemma 3} shows that the corresponding Legendre-Fenchel transform of post-retirement value function $\tilde{U}(z)$ is
	\begin{equation*}
		\tilde{U}(z)=B_{2,\scriptscriptstyle PR}z^{n_{2}}+\frac{1-\delta(1-k)}{\delta(1-k)}K_{1}\bar{L}^{\frac{(1-k)(1-\delta)}{1-\delta(1-k)}}z^{\frac{\delta(1-k)}{\delta(1-k)-1}}-\frac{d}{r}z.
	\end{equation*}
	Meanwhile, the dual transform involving the pre-retirement part $\tilde{u}(z)$ stays the same; hence Equation $(V1)$ from (\ref{Equation 9}) takes the identical solution.
	Afterwards, using the smooth fit conditions at $z\!=\!\bar{z}$ and $z\!=\!\tilde{y}$, we construct a four-equations system to achieve the unknowns, $B_{11}$, $B_{21}$, $B_{12}$ and $\bar{z}$.
	\begin{itemize}
		\item $\mathcal{C}^{0}$ condition at $z=\bar{z}$
		\begin{equation*}
			B_{11}\bar{z}^{n_{1}}\!+\!B_{21}\bar{z}^{n_{2}}\!+\!\frac{A_{1}}{\Gamma_{1}}\bar{z}^{\frac{\delta(1\!-\!k)}{\delta(1\!-\!k)\!-\!1}}\!+\!\frac{w(\bar{L}\!-\!L)}{r}\bar{z}\!=\!B_{2,\scriptscriptstyle PR}\bar{z}^{n_{2}}\!+\!\frac{1\!-\!\delta(1\!-\!k)}{\delta(1\!-\!k)}K_{1}\bar{L}^{\frac{(1\!-\!k)(1\!-\!\delta)}{1\!-\!\delta(1\!-\!k)}}\bar{z}^{\frac{\delta(1\!-\!k)}{\delta(1\!-\!k)\!-\!1}};
		\end{equation*}
		\item $\mathcal{C}^{1}$ condition at $z=\bar{z}$
		\begin{equation*}
			n_{1}B_{11}\bar{z}^{n_{1}\!-\!1}\!+\!n_{2}B_{21}\bar{z}^{n_{2}\!-\!1}\!+\!\frac{\delta(1\!-\!k)}{\delta(1\!-\!k)\!-\!1}\!\frac{A_{1}}{\Gamma_{1}}\!\bar{z}^{\frac{1}{\delta(1\!-\!k)\!-\!1}}\!+\!\frac{w(\bar{L}\!-\!L)}{r}\!=\!
			n_{2}B_{2,\scriptscriptstyle PR}\bar{z}^{n_{2}\!-\!1}
			\!-\!K_{1}\!\bar{L}^{\frac{(1\!-\!k)(1\!-\!\delta)}{1\!-\!\delta(1\!-\!k)}}\!\bar{z}^{\frac{1}{\delta(1\!-\!k)\!-\!1}};
		\end{equation*}
		\item $\mathcal{C}^{0}$ condition at $z=\tilde{y}$
		\begin{equation*}
			B_{11}\tilde{y}^{n_{1}}+B_{21}\tilde{y}^{n_{2}}+\frac{A_{1}}{\Gamma_{1}}\tilde{y}^{\frac{\delta(1-k)}{\delta(1-k)-1}}-\frac{w L}{r}\tilde{y}=B_{12}\tilde{y}^{n_{1}}+\frac{A_{2}}{\Gamma_{2}}\tilde{y}^{-\frac{1-k}{k}};
		\end{equation*}
		\item $\mathcal{C}^{1}$ condition at $z=\tilde{y}$
		\begin{equation*}
			n_{1}B_{11}\tilde{y}^{n_{1}-1}+n_{2}B_{21}\tilde{y}^{n_{2}-1}+\frac{\delta(1-k)}{\delta(1-k)-1}\frac{A_{1}}{\Gamma_{1}}\tilde{y}^{\frac{1}{\delta(1-k)-1}}-\frac{w L}{r}= n_{1}B_{12}\tilde{y}^{n_{1}-1}-\frac{1-k}{k}\frac{A_{2}}{\Gamma_{2}}\tilde{y}^{-\frac{1}{k}}.
		\end{equation*}
	\end{itemize}

   \noindent After obtaining the closed forms of $v(z)$ separately in Case 1 and Case 2, and given the initial wealth $x\ge R_{pre}$, the optimal Lagrange multiplier $\lambda^{*}$ can be acquired through solving the equation $x=-v^{\prime}(\lambda^{*})$, due to the fact that
   $V(x)=\inf\limits_{\lambda>0}\big[\tilde{V}(\lambda)+\lambda x\big]=\inf\limits_{\lambda>0}\big[v(\lambda)+\lambda x\big]=v(\lambda^*)+\lambda^* x$ holds under the differentiable property of $v(\cdot)$. Then the optimal dual process of wealth follows $Z^{*}(t)=\lambda^{*}e^{\gamma t}H(t)$.
   
   	\begin{proposition}\label{Proposition Appendix 5}
   	For Case 1 and Case 2, the optimal retirement time is $\tau^{*}=\inf\limits_{t\ge 0} \{Z^{*}(t)\leq\bar{z}\}$, the optimal consumption-portfolio-leisure plan $\{c^{*}(t),\pi^{*}(t),l^{*}(t)\}$ before retirement is given by
   	\begin{equation*}
   		c^{*}(t)=\begin{cases}
   			L^{-\frac{(1-k)(1-\delta)}{\delta(1-k)-1}}(Z^{*}(t))^{\frac{1}{\delta(1-k)-1}}, &\bar{z}<Z^{*}(t)<\tilde{y},\\
   			\left(\frac{1-\delta}{\delta w}\right)^{\frac{(1-\delta)(1-k)}{k}}(Z^{*}(t))^{-\frac{1}{k}}, & Z^{*}(t)\ge\tilde{y},\\
   		\end{cases}
   	\end{equation*}
   	\begin{equation*}
   		l^{*}(t)=\begin{cases}
   			L, & \bar{z}<Z^{*}(t)<\tilde{y},\\
   			\left(\frac{1-\delta}{\delta w}\right)^{-\frac{\delta(1-k)-1}{k}}(Z^{*}(t))^{-\frac{1}{k}}, & Z^{*}(t)\ge\tilde{y},\\
   		\end{cases}
   	\end{equation*}
   	\begin{equation*}
   		\pi^{*}(t)=
   		\begin{cases}
   			\frac{\theta}{\sigma}\left[n_{1}(n_{1}-1)B_{11}(Z^{*}(t))^{n_{1}-1}+n_{2}(n_{2}-1)B_{21}(Z^{*}(t))^{n_{2}-1}\right.& \\
   			\qquad\qquad\qquad\qquad\left.+\frac{\delta(1-k)}{(\delta(1-k)-1)^{2}}\frac{A_{1}}{\Gamma_{1}}(Z^{*}(t))^{\frac{1}{\delta(1-k)-1}}\right],	& \bar{z}<Z^{*}(t)<\tilde{y},\\
   			\frac{\theta}{\sigma}\left[n_{1}(n_{1}\!-\!1)B_{12}(Z^{*}(t))^{n_{1}-1}+\frac{1-k}{k^{2}}\frac{A_{2}}{\Gamma_{2}}\right.(Z^{*}\left.(t))^{-\frac{1}{k}}\right], & Z^{*}(t)\ge\tilde{y}.
   		\end{cases}
   	\end{equation*}
   \end{proposition}
   \begin{proof}
   	The optimal consumption and leisure strategies come from \cite[Lemma 2.1]{2021arXiv210615426D}, and the optimal portfolio strategy is derived by $\pi^{*}(t)=\frac{\theta}{\sigma}Z^{*}(t)v^{\prime\prime}(Z^{*}(t))$ from \cite[Section 5, Theorem 3]{he1993labor}.
   \end{proof}

	\section{Calculation of Variational Inequalities (\ref{Equation 12})}\label{Appendix 8}
	
	\noindent Recalling the condition $\bar{z}<\tilde{y}$ in Lemma \ref{Lemma 5}, the problem to be solved is split into four different cases depending on the relationship between $\tilde{y}$ with $\hat{z}$, and $R_{post}$ with $\frac{d}{r}$. We provide a diagram for a clear classification.
		\begin{center}
			\begin{tikzpicture}[
				grow=right,
				level 1/.style={sibling distance=1.5cm,level distance=6.5cm},
				level 2/.style={sibling distance=1.5cm, level distance=5.5cm},
				edge from parent/.style={very thick,draw=blue!40!black!60,
					shorten >=5pt, shorten <=5pt},
				edge from parent path={(\tikzparentnode.east) -- (\tikzchildnode.west)},
				kant/.style={text width=2cm, text centered, sloped},
				every node/.style={text ragged, inner sep=2mm},
				punkt/.style={rectangle, rounded corners, shade, top color=white,
					bottom color=blue!50!black!20, draw=blue!40!black!60, very
					thick }
				]
				
				\node[punkt, text width=6.5em] {$0<\bar{z}<\tilde{y}\leq\hat{z}$}
				child{ node[punkt, text width=22em] {Case 4. $0<\bar{z}<\tilde{y}\leq\hat{z}$, $R_{pre}>\frac{d-w\bar{L}}{r}$ $\&$ $R_{post}>\frac{d}{r}$}
					edge from parent{
						node[kant, above] {}}
				}
				child { node[punkt, text width=22em] {Case 3. $0<\bar{z}<\tilde{y}\leq\hat{z}$, $R_{pre}>\frac{d-w\bar{L}}{r}$ $\&$ $R_{post}= \frac{d}{r}$}
					edge from parent{
						node[kant, above] {}}
				};
			\end{tikzpicture}
		\end{center}
		
		\begin{center}
			\begin{tikzpicture}[
				grow=right,
				level 1/.style={sibling distance=1.5cm,level distance=6.5cm},
				level 2/.style={sibling distance=1.5cm, level distance=5.5cm},
				edge from parent/.style={very thick,draw=blue!40!black!60,
					shorten >=5pt, shorten <=5pt},
				edge from parent path={(\tikzparentnode.east) -- (\tikzchildnode.west)},
				kant/.style={text width=2cm, text centered, sloped},
				every node/.style={text ragged, inner sep=2mm},
				punkt/.style={rectangle, rounded corners, shade, top color=white,
					bottom color=blue!50!black!20, draw=blue!40!black!60, very
					thick }
				]
				
				\node[punkt, text width=6.5em] {$0<\bar{z}<\hat{z}<\tilde{y}$}
				child{ node[punkt, text width=22em] {Case 6. $0<\bar{z}<\hat{z}<\tilde{y}$, $R_{pre}>\frac{d-w\bar{L}}{r}$ $\&$ $R_{post}>\frac{d}{r}$}
					edge from parent{
						node[kant, above] {}}
				}
				child { node[punkt, text width=22em] {Case 5. $0<\bar{z}<\hat{z}<\tilde{y}$, $R_{pre}>\frac{d-w\bar{L}}{r}$ $\&$ $R_{post}= \frac{d}{r}$}
					edge from parent{
						node[kant, above] {}}
				};
			\end{tikzpicture}
		\end{center}
		
	 \noindent Additionally, we assume that ${\phi}(t,z)$ takes the time-separated form, ${\phi}(t,z)\!=\!e^{\!-\!\gamma t}v(z)$, as in \cite[Appendix A]{choi2008optimal}.
	 
	 \paragraph{Case 3. $0<\bar{z}<\tilde{y}\leq\hat{z}$, $R_{pre}>\frac{d-w\bar{L}}{r}$ $\&$ $R_{post}=\frac{d}{r}$}~{}
	\newline
	We begin with the condition $(V3)$ in (\ref{Equation 12}), the following differential equation is obtained,
	\begin{equation*}
		-\gamma v(z)+(\gamma-r)z v^{\prime}(z)+\frac{1}{2}\theta^{2}z^{2}v^{\prime\prime}(z)+\tilde{u}(z)-(d-w\bar{L})z=0,
	\end{equation*}
	which is identical with Equation (\ref{Equation Appendix 5}), hence shares the same solution as
	\begin{equation}
		v(z)=\begin{cases}
			B_{11}z^{n_{1}}+B_{21}z^{n_{2}}+\frac{A_{1}}{\Gamma_{1}}z^{\frac{\delta(1-k)}{\delta(1-k)-1}}+\frac{w(\bar{L}-L)-d}{r}z, & \bar{z}< z<\tilde{y},\\
			
			B_{12}z^{n_{1}}+B_{22}z^{n_{2}}+\frac{A_{2}}{\Gamma_{2}}z^{-\frac{1-k}{k}}+\frac{w\bar{L}-d}{r}z,& \tilde{y}\leq z<\hat{z}.\label{Equation Appendix 6}
		\end{cases}
	\end{equation}
	As follows, a six-equations system is established to obtain the unknown parameters $B_{11}$, $B_{21}$, $B_{12}$, $B_{22}$, $\bar{z}$ and $\hat{z}$.  
	\begin{itemize}
		\item $\mathcal{C}^{0}$ condition at $z=\bar{z}$
		\begin{equation*}
			B_{11}\bar{z}^{n_{1}}+B_{21}\bar{z}^{n_{2}}+\frac{A_{1}}{\Gamma_{1}}\bar{z}^{\frac{\delta(1-k)}{\delta(1-k)-1}}+\frac{w(\bar{L}-L)}{r}\bar{z}=\frac{1-\delta(1-k)}{\delta(1-k)}K_{1}\bar{L}^{\frac{(1-k)(1-\delta)}{1-\delta(1-k)}}\bar{z}^{\frac{\delta(1-k)}{\delta(1-k)-1}}.
		\end{equation*}
		\item $\mathcal{C}^{1}$ condition at $z=\bar{z}$
		\begin{equation*}
			n_{1}B_{11}\bar{z}^{n_{1}-1}\!+\!n_{2}B_{21}\bar{z}^{n_{2}-1}\!+\!\frac{\delta(1-k)}{\delta(1-k)-1}\frac{A_{1}}{\Gamma_{1}}\bar{z}^{\frac{1}{\delta(1-k)-1}}\!+\!\frac{w(\bar{L}-L)}{r}\!=\!-K_{1}\bar{L}^{\frac{(1-k)(1-\delta)}{1-\delta(1-k)}}\bar{z}^{\frac{1}{\delta(1-k)-1}}.
		\end{equation*}
		\item $\mathcal{C}^{0}$ condition at $z=\tilde{y}$
		\begin{equation*}
			B_{11}\tilde{y}^{n_{1}}+B_{21}\tilde{y}^{n_{2}}+\frac{A_{1}}{\Gamma_{1}}\tilde{y}^{\frac{\delta(1-k)}{\delta(1-k)-1}}-\frac{w L}{r}\tilde{y}=B_{12}\tilde{y}^{n_{1}}+B_{22}\tilde{y}^{n_{2}}+\frac{A_{2}}{\Gamma_{2}}\tilde{y}^{-\frac{1-k}{k}}.
		\end{equation*}
		\item $\mathcal{C}^{1}$ condition at $z=\tilde{y}$
		\begin{equation*}
			n_{1}B_{11}\tilde{y}^{n_{1}\!-\!1}\!+\!n_{2}B_{21}\tilde{y}^{n_{2}\!-\!1}\!+\!\frac{\delta(1\!-\!k)}{\delta(1\!-\!k)\!-\!1}\!\frac{A_{1}}{\Gamma_{1}}\tilde{y}^{\frac{1}{\delta(1\!-\!k)\!-\!1}}\!-\!\frac{w L}{r}\!=\! n_{1}B_{12}\tilde{y}^{n_{1}\!-\!1}\!+\!n_{2}B_{22}\tilde{y}^{n_{2}\!-\!1}\!-\!\frac{1\!-\!k}{k}\!\frac{A_{2}}{\Gamma_{2}}\tilde{y}^{\!-\!\frac{1}{k}}.
		\end{equation*}
		\item $\mathcal{C}^{1}$ condition at $z=\hat{z}$
		\begin{equation*}
			n_{1}B_{12}\hat{z}^{n_{1}-1}+n_{2}B_{22}\hat{z}^{n_{2}-1}-\frac{1-k}{k}\frac{A_{2}}{\Gamma_{2}}\hat{z}^{-\frac{1}{k}}+\frac{w\bar{L}-d}{r}+R_{pre}=0.
		\end{equation*}
		\item $\mathcal{C}^{2}$ condition at $z=\hat{z}$
		\begin{equation*}
			n_{1}(n_{1}-1)B_{12}\hat{z}^{n_{1}-2}+n_{2}(n_{2}-1)B_{22}\hat{z}^{n_{2}-2}+\frac{1-k}{k^{2}}\frac{A_{2}}{\Gamma_{2}}\hat{z}^{-\frac{1+k}{k}}=0.
		\end{equation*}	
	\end{itemize}
	
	\paragraph{Case 4. $0<\bar{z}<\tilde{y}\leq\hat{z}$, $R_{pre}>\frac{d-w\bar{L}}{r}$ $\&$ $R_{post}> \frac{d}{r}$}~{}
	\newline
	The only difference between this case and the previous
	one occurs in $z=\bar{z}$. Since $R_{post}>\frac{d}{r}$, Lemma \ref{Lemma 3} shows that the Legendre-Fenchel transform of post-retirement value function $\tilde{U}(\bar{z})$ is
	\begin{equation*}
		\tilde{U}(\bar{z})=B_{2,\scriptscriptstyle PR}\bar{z}^{n_{2}}+\frac{1-\delta(1-k)}{\delta(1-k)}K_{1}\bar{L}^{\frac{(1-k)(1-\delta)}{1-\delta(1-k)}}\bar{z}^{\frac{\delta(1-k)}{\delta(1-k)-1}}-\frac{d}{r}\bar{z}.
	\end{equation*}
	As the same before, we set up a six-equation system to achieve the unknowns, $B_{11}$, $B_{21}$, $B_{12}$, $B_{22}$, $\bar{z}$ and $\hat{z}$. Compared with the first case, only $\mathcal{C}^{0}$ and $\mathcal{C}^{1}$ conditions at $z=\bar{z}$ change, whereas all the others keep true.
	\begin{itemize}
		\item $\mathcal{C}^{0}$ condition at $z=\bar{z}$
		\begin{equation*}
			B_{11}\bar{z}^{n_{1}}\!+\!B_{21}\bar{z}^{n_{2}}\!+\!\frac{A_{1}}{\Gamma_{1}}\bar{z}^{\frac{\delta(1\!-\!k)}{\delta(1\!-\!k)\!-\!1}}\!+\!\frac{w(\bar{L}\!-\!L)}{r}\bar{z}\!=\!B_{2,\scriptscriptstyle PR}\bar{z}^{n_{2}}\!+\!\frac{1\!-\!\delta(1\!-\!k)}{\delta(1\!-\!k)}K_{1}\bar{L}^{\frac{(1\!-\!k)(1\!-\!\delta)}{1\!-\!\delta(1\!-\!k)}}\bar{z}^{\frac{\delta(1\!-\!k)}{\delta(1\!-\!k)\!-\!1}}.
		\end{equation*}
		\item $\mathcal{C}^{1}$ condition at $z=\bar{z}$
		\begin{equation*}
			n_{1}\!B_{11}\!\bar{z}^{n_{1}\!-\!1}\!+\!n_{2}\!B_{21}\!\bar{z}^{n_{2}\!-\!1}\!+\!\frac{\delta(1\!-\!k)}{\delta(1\!-\!k)\!-\!1}\!\frac{A_{1}}{\Gamma_{1}}\!\bar{z}^{\frac{1}{\delta(1\!-\!k)\!-\!1}}\!+\!\frac{w(\bar{L}\!-\!L)}{r}\!=\!
			n_{2}\!B_{2,\scriptscriptstyle PR}\!\bar{z}^{n_{2}\!-\!1}
			\!-\!K_{1}\!\bar{L}^{\frac{(1\!-\!k)(1\!-\!\delta)}{1\!-\!\delta(1\!-\!k)}}\!\bar{z}^{\frac{1}{\delta(1\!-\!k)\!-\!1}}.
		\end{equation*}
	\end{itemize}
	
	\paragraph{Case 5. $0<\bar{z}<\hat{z}<\tilde{y}$, $R_{pre}>\frac{d-w\bar{L}}{r}$ $\&$ $R_{post}=\frac{d}{r}$}~{}
	\newline
	Firstly, the interval $0<\bar{z}< z<\hat{z}<\tilde{y}$, where the condition $(V3)$ of (\ref{Equation 12}) holds, is considered. Also adopting the time-independent form of ${\phi}(t,z)=e^{-\gamma t}v(z)$, the following differential equation is obtained $-\gamma v(z)+(\gamma-r)z v^{\prime}(z)+\frac{1}{2}\theta^{2}z^{2}v^{\prime\prime}(z)+\tilde{u}(z)-(d-w\bar{L})z=0$.
	The dual transform of $u(c,l)$ is $\tilde{u}(z)=A_{1}z^{\frac{\delta(1-k)}{\delta(1-k)-1}}-w L z$ in the considered interval; therefore, the above differential equation takes the identical form of the one in $0<\bar{z}< z<\tilde{y}<\hat{z}$ of Case 3. The solution of $v(z)$ is given directly from (\ref{Equation Appendix 6}), only changing the parameters' notations from $B_{11}$ to $B_{1}$ and $B_{21}$ to $B_{2}$ respectively,
	\begin{equation*}
		v(z)=B_{1}z^{n_{1}}+B_{2}z^{n_{2}}+\frac{A_{1}}{\Gamma_{1}}z^{\frac{\delta(1-k)}{\delta(1-k)-1}}+\frac{w(\bar{L}-L)-d}{r}z, \quad\bar{z}< z<\hat{z}.
	\end{equation*}
	Next, a four-equations system is set up to derive the desired parameters $B_{1}$, $B_{2}$, $\bar{z}$, $\hat{z}$. The same arguments with Case 3, only $\mathcal{C}^{1}$ and $\mathcal{C}^{2}$ conditions in $z=\hat{z}$ changes.
	\begin{itemize}
		\item $\mathcal{C}^{0}$ condition at $z=\bar{z}$
		\begin{equation*}
			B_{1}\bar{z}^{n_{1}}+B_{2}\bar{z}^{n_{2}}+\frac{A_{1}}{\Gamma_{1}}\bar{z}^{\frac{\delta(1-k)}{\delta(1-k)-1}}+\frac{w(\bar{L}-L)}{r}\bar{z}=\frac{1-\delta(1-k)}{\delta(1-k)}K_{1}\bar{L}^{\frac{(1-k)(1-\delta)}{1-\delta(1-k)}}\bar{z}^{\frac{\delta(1-k)}{\delta(1-k)-1}}.
		\end{equation*}
		\item $\mathcal{C}^{1}$ condition at $z=\bar{z}$
		\begin{equation*}
			n_{1}B_{1}\bar{z}^{n_{1}-1}+n_{2}B_{2}\bar{z}^{n_{2}-1}+\frac{\delta(1-k)}{\delta(1-k)-1}\frac{A_{1}}{\Gamma_{1}}\bar{z}^{\frac{1}{\delta(1-k)-1}}+\frac{w(\bar{L}-L)}{r}=-K_{1}\bar{L}^{\frac{(1-k)(1-\delta)}{1-\delta(1-k)}}\bar{z}^{\frac{1}{\delta(1-k)-1}}.
		\end{equation*}
		\item $\mathcal{C}^{1}$ condition at $z=\hat{z}$
		\begin{equation*}
			n_{1}B_{1}\hat{z}^{n_{1}-1}+n_{2}B_{2}\hat{z}^{n_{2}-1}+\frac{\delta(1-k)}{\delta(1-k)-1}\frac{A_{1}}{\Gamma_{1}}\hat{z}^{\frac{1}{\delta(1-k)-1}}+\frac{w(\bar{L}-L)-d}{r}+R_{pre}=0.
		\end{equation*}
		\item $\mathcal{C}^{2}$ condition at $z=\hat{z}$
		\begin{equation*}
			n_{1}(n_{1}-1)B_{1}\hat{z}^{n_{1}-2}+n_{2}(n_{2}-1)B_{2}\hat{z}^{n_{2}-2}+\frac{\delta(1-k)}{(\delta(1-k)-1)^{2}}\frac{A_{1}}{\Gamma_{1}}\hat{z}^{\frac{2-\delta(1-k)}{\delta(1-k)-1}}=0.
		\end{equation*}
	\end{itemize}
	
	\paragraph{Case 6. $0<\bar{z}<\hat{z}<\tilde{y}$, $R_{pre}>\frac{d-w\bar{L}}{r}$ $\&$ $R_{post}>\frac{d}{r}$}~{}
	\newline
	We now move to Case 6. The only difference from the previous case happens on the condition $R_{post}>\frac{d}{r}$, which is mainly involved in the post-retirement part; hence, the solution of the partial differential equation corresponding to Condition $(V3)$ in (\ref{Equation 12}) remains unchanged, that is,
	\begin{equation*}
		v(z)=B_{1}z^{n_{1}}+B_{2}z^{n_{2}}+\frac{A_{1}}{\Gamma_{1}}z^{\frac{\delta(1-k)}{\delta(1-k)-1}}+\frac{w(\bar{L}-L)-d}{r}z, \quad \bar{z}< z<\hat{z}.
	\end{equation*}
	Considering the smooth fit conditions at $\bar{z}$ and $\hat{z}$, we construct a four-equations system to deduce the values of unknown parameters $B_{1}$, $B_{2}$, $\bar{z}$ and $\hat{z}$.
	\begin{itemize}
		\item $\mathcal{C}^{0}$ condition at $z=\bar{z}$
		\begin{equation*}
			B_{1}\bar{z}^{n_{1}}+B_{2}\bar{z}^{n_{2}}+\frac{A_{1}}{\Gamma_{1}}\bar{z}^{\frac{\delta(1-k)}{\delta(1-k)-1}}+\frac{w(\bar{L}-L)}{r}\bar{z}=B_{2,\scriptscriptstyle PR}\bar{z}^{n_{2}}+\frac{1-\delta(1-k)}{\delta(1-k)}K_{1}\bar{L}^{\frac{(1-k)(1-\delta)}{1-\delta(1-k)}}\bar{z}^{\frac{\delta(1-k)}{\delta(1-k)-1}}.
		\end{equation*}
		\item $\mathcal{C}^{1}$ condition at $z=\bar{z}$
		\begin{equation*}
			n_{1}B_{1}\bar{z}^{n_{1}\!-\!1}\!+\!n_{2}B_{2}\bar{z}^{n_{2}\!-\!1}\!+\!\frac{\delta(1\!-\!k)}{\delta(1\!-\!k)\!-\!1}\!\frac{A_{1}}{\Gamma_{1}}\!\bar{z}^{\frac{1}{\delta(1\!-\!k)\!-\!1}}\!+\!\frac{w(\bar{L}\!-\!L)}{r}\!=\!
			n_{2}B_{2,\scriptscriptstyle PR}\bar{z}^{n_{2}\!-\!1}
			\!-\!K_{1}\!\bar{L}^{\frac{(1\!-\!k)(1\!-\!\delta)}{1\!-\!\delta(1\!-\!k)}}\!\bar{z}^{\frac{1}{\delta(1\!-\!k)\!-\!1}}.
		\end{equation*}
		\item $\mathcal{C}^{1}$ condition at $z=\hat{z}$
		\begin{equation*}
			n_{1}B_{1}\hat{z}^{n_{1}-1}+n_{2}B_{2}\hat{z}^{n_{2}-1}+\frac{\delta(1-k)}{\delta(1-k)-1}\frac{A_{1}}{\Gamma_{1}}\hat{z}^{\frac{1}{\delta(1-k)-1}}+\frac{w(\bar{L}-L)-d}{r}+R_{pre}=0.
		\end{equation*}
		\item $\mathcal{C}^{2}$ condition at $z=\hat{z}$
		\begin{equation*}
			n_{1}(n_{1}-1)B_{1}\hat{z}^{n_{1}-2}+n_{2}(n_{2}-1)B_{2}\hat{z}^{n_{2}-2}+\frac{\delta(1-k)}{(\delta(1-k)-1)^{2}}\frac{A_{1}}{\Gamma_{1}}\hat{z}^{\frac{2-\delta(1-k)}{\delta(1-k)-1}}=0.
		\end{equation*}
	\end{itemize}
    
    \noindent Same argument with Case 1 and Case 2 in Appendix \ref{Appendix 7}, given the initial wealth $x\ge R_{pre}$ and solving $x=-v^{\prime}(\lambda^{*})$, we can obtain the optimal Lagrange multiplier $\lambda^{*}$ and then the optimal process $Z^{*}(t)=\lambda^{*}e^{\gamma t}H(t)$. 
    
    \begin{proposition}\label{Proposition Appendix 6}
    Under the condition $\tilde{y}\leq\hat{z}$, corresponding to Case 3 and Case 4, the optimal consumption-portfolio-leisure plan $\{c^{*}(t),\pi^{*}(t),l^{*}(t)\}$ before retirement is given by
    	\begin{equation*}
    		c^{*}(t)=\begin{cases}
    			\left(\frac{1-\delta}{\delta w}\right)^{\frac{(1-\delta)(1-k)}{k}}(Z^{*}(t))^{-\frac{1}{k}}, & \tilde{y}\leq Z^{*}(t)\leq\hat{z},\\
    			L^{-\frac{(1-k)(1-\delta)}{\delta(1-k)-1}}(Z^{*}(t))^{\frac{1}{\delta(1-k)-1}}, &\bar{z}<Z^{*}(t)<\tilde{y},
    		\end{cases}
    	\end{equation*}
    	\begin{equation*}
    		l^{*}(t)=\begin{cases}
    			\left(\frac{1-\delta}{\delta w}\right)^{-\frac{\delta(1-k)-1}{k}}(Z^{*}(t))^{-\frac{1}{k}}, & \tilde{y}\leq Z^{*}(t)\leq\hat{z},\\
    			L, & \bar{z}<Z^{*}(t)<\tilde{y},
    		\end{cases}
    	\end{equation*}
    	\begin{equation*}
    		\pi^{*}(t)=
    		\begin{cases}
    			\frac{\theta}{\sigma}\bigg[n_{1}(n_{1}-1)B_{12}(Z^{*}(t))^{n_{1}-1}+n_{2}(n_{2}-1)B_{22}(Z^{*}(t))^{n_{2}-1} & \\
    			\qquad\qquad\qquad\qquad\qquad\qquad\qquad\qquad+\frac{1-k}{k^{2}}\frac{A_{2}}{\Gamma_{2}}(Z^{*}(t))^{-\frac{1}{k}}\bigg],& \tilde{y}\leq Z^{*}(t)\leq\hat{z},\\
    			\frac{\theta}{\sigma}\bigg[n_{1}(n_{1}-1)B_{11}(Z^{*}(t))^{n_{1}-1}+n_{2}(n_{2}-1)B_{21}(Z^{*}(t))^{n_{2}-1} & \\
    			\qquad\qquad\qquad\qquad\qquad\quad+\frac{\delta(1-k)}{(\delta(1-k)-1)^{2}}\frac{A_{1}}{\Gamma_{1}}(Z^{*}(t))^{\frac{1}{\delta(1-k)-1}}\bigg], & \bar{z}<Z^{*}(t)<\tilde{y}.
    		\end{cases}
    	\end{equation*}
    Meanwhile, under the condition $\hat{z}<\tilde{y}$, corresponding to Case 5 and Case 6, the optimal consumption-portfolio-leisure plan $\{c^{*}(t),\pi^{*}(t),l^{*}(t)\}$ before retirement is given by
    \begin{equation*}
    	c^{*}(t)=L^{-\frac{(1-k)(1-\delta)}{\delta(1-k)-1}}\left(Z^{*}(t)\right)^{\frac{1}{\delta(1-k)-1}},\qquad l^{*}(t)=L,
    \end{equation*}
    \begin{equation*}
    	\pi^{*}(t)\!=\!
    	\frac{\theta}{\sigma}\!\left[\!n_{1}(n_{1}\!-\!1)B_{1}(Z^{*}(t))^{n_{1}\!-\!1}\!+\!n_{2}(n_{2}\!-\!1)B_{2}(Z^{*}(t))^{n_{2}\!-\!1}\!+\!\frac{\delta(1\!-\!k)}{(\delta(1\!-\!k)\!-\!1)^{2}}\frac{A_{1}}{\Gamma_{1}}(Z^{*}(t))^{\frac{1}{\delta(1\!-\!k)\!-\!1}}\right].
    \end{equation*}
    \end{proposition}
    \begin{proof}
    	Follow the lines of Proposition \ref{Proposition Appendix 5}.
    \end{proof}

\end{document}